\documentclass[12pt]{article}
\usepackage{amsmath}
\usepackage{graphicx}
\usepackage{enumerate}
\usepackage{natbib}
\usepackage{url} 
\usepackage{authblk}
\usepackage{subcaption}
\usepackage{hyperref}
\usepackage{times}
\usepackage{latexsym}
\usepackage[T1]{fontenc}
\usepackage{booktabs}
\usepackage{nicefrac}
\usepackage{bm}
\usepackage{graphicx}
\graphicspath{ {figs/} } 
\usepackage{amssymb}
\usepackage{amsthm}
\usepackage{xcolor}

\usepackage{multibib}

\newcites{app}{References}

\usepackage{lipsum}
\usepackage{listings}
\usepackage{float}
\usepackage{tikz}
\usetikzlibrary{positioning} 
\usepackage{tcolorbox}
\usepackage{caption}
\usepackage{fancyvrb}
\usepackage[english]{babel}
\definecolor{codegreen}{rgb}{0,0.6,0}
\definecolor{codegray}{rgb}{0.5,0.5,0.5}
\definecolor{codepurple}{rgb}{0.58,0,0.82}
\definecolor{backcolour}{rgb}{0.95,0.95,0.92}

\lstdefinestyle{mystyle}{
    backgroundcolor=\color{backcolour},   
    commentstyle=\color{codegreen},
    keywordstyle=\color{magenta},
    numberstyle=\tiny\color{codegray},
    stringstyle=\color{codepurple},
    basicstyle=\ttfamily\footnotesize,
    breakatwhitespace=false,         
    breaklines=falstrue,                 
    captionpos=b,    
    language=R,
    keepspaces=true,                 
    numbers=left,                    
    numbersep=5pt,                  
    showspaces=false,                
    showstringspaces=false,
    showtabs=falsfalse,                  
    tabsize=2
}
\newcommand{\eps}{\varepsilon}

\newcommand*{\bs}{\boldsymbol}
\newcommand*{\mc}{\mathcal}
\newcommand*{\bb}{\mathbb}
\newcommand*{\PP}{\mathbb{P}}

\newcommand*{\E}{\mathbb{E}}

\newcommand{\mb}{\mathbf}
\newcommand*{\mbX}{\mathbf{X}}

\newcommand*{\mbH}{\mathbf{H}}

\newcommand*{\mbW}{\mathbf{W}}
\newcommand*{\mbY}{\mathbf{Y}}
\newcommand*{\mbZ}{\mathbf{Z}}

\newcommand*{\mbh}{\mathbf{h}}

\newcommand*{\mbw}{\mathbf{w}}
\newcommand*{\mbx}{\mathbf{x}}

\newcommand*{\mbz}{\mathbf{z}}

\newcommand{\qestimand}{g^{\tau}_{\textrm{QPE}(\mb X \to \mbY)}(\mb x,\mbw)}
\newcommand{\qest}{\hat{g}^{nm,\tau}_{\textrm{QPE}(\mb X \to \mbY)}(\bs{X}_n^{m_n}, \bs{Y}_n^{m_n}, \bs{W}_n^{m_n})(\mb x,\mbw)}

\newcommand{\indep}{\perp\!\!\!\!\perp} 

\DeclareMathOperator*{\argmin}{arg\,min}

\newtheorem{theorem}{Theorem}[section]
\newtheorem{lemma}[theorem]{Lemma}
\newtheorem{assumption}{Assumption}
\newtheorem{corollary}{Corollary}[theorem]
\theoremstyle{definition}
\newtheorem{definition}[theorem]{Definition}
\newtheorem{example}[theorem]{Example}

\theoremstyle{remark}
\newtheorem{remark}[theorem]{Remark}

\numberwithin{equation}{section}

\tcbuselibrary{skins,breakable}
\usetikzlibrary{shadings,shadows}
{\endtcolorbox}

\newcommand{\blind}{1}

\begin{document}

\def\spacingset#1{\renewcommand{\baselinestretch}%
{#1}\small\normalsize} \spacingset{1}


\if1\blind
{
  \title{\bf Quantile-based causal inference for spatio-temporal processes: Assessing the impacts of wildfires on US air quality}
  \author[a]{Zipei Geng}
  \author[b]{Jordan Richards}
  \author[a]{Rapha\"{e}l Huser}
  \author[a]{Marc G. Genton}
  
  \affil[a]{Statistics Program, King Abdullah University of Science and Technology, Thuwal, Saudi Arabia}
  \affil[b]{School of Mathematics and Maxwell Institute for Mathematical Sciences, University of Edinburgh, Edinburgh, EH9 3FD, UK
}
  \date{}
  \maketitle
} \fi

\if0\blind
{
  \bigskip
  \bigskip
  \bigskip
  \begin{center}
    {\LARGE\bf Quantile-based causal inference for spatio-temporal processes: Assessing the impacts of wildfires on US air quality}
\end{center}
  \medskip
} \fi

\bigskip
\begin{abstract}
 Wildfires pose an increasingly severe threat to air quality, yet quantifying their causal impact remains challenging due to unmeasured meteorological and geographic confounders. Moreover, wildfire impacts on air quality may exhibit heterogeneous effects across pollution levels, which conventional mean-based causal methods fail to capture. To address these challenges, we develop a Quantile-based Latent Spatial Confounder Model (QLSCM) that substitutes conditional expectations with conditional quantiles, enabling causal analysis across the entire outcome distribution. We establish the causal interpretation of QLSCM theoretically, prove the identifiability of causal effects, and demonstrate estimator consistency under mild conditions. Simulations confirm the bias correction capability and the advantage of quantile-based inference over mean-based approaches. Applying our method to contiguous US wildfire and air quality data, we uncover important heterogeneous effects: fire radiative power exerts significant positive causal effects on aerosol optical depth at high quantiles in Western states like California and Oregon, while insignificant at lower quantiles. This indicates that wildfire impacts on air quality primarily manifest during extreme pollution events. Regional analyses reveal that Western and Northwestern regions experience the strongest causal effects during such extremes. These findings provide critical insights for environmental policy by identifying where and when mitigation efforts would be most effective.
\end{abstract}

\noindent%
\vfill

\allowdisplaybreaks

\newpage
\spacingset{1.9} 
\baselineskip=26pt
\section{Introduction}
\label{sec:intro}

\subsection{Air Quality and Wildfire Intensity in the US}
\begin{hyphenrules}{nohyphenation}
 Analyzing air quality in relation to wildfires is essential, as wildfires emit large quantities of pollutants, such as particulate matter (PM), carbon monoxide (CO), volatile organic compounds, and greenhouse gases, which impact human health and the environment \citep{reid2016critical}. Aerosol optical depth (AOD), a measure of air quality, is strongly associated with lung cancer incidences \citep{yu2024can} and wildfire smoke \citep{zielinski2024impact}, with the latter linked to increased rates of respiratory and cardiovascular diseases, especially in vulnerable populations \citep{holm2021health}. Furthermore, particulate emissions serve as significant contributors to ongoing climate change, subsequently altering future fire patterns, and thus establishing a complex feedback mechanism connecting climatic conditions, atmospheric composition, and wildfire behavior \citep{westerling2008climate}. \citet{yang2023data} demonstrated significant wildfire effects on fine PM concentrations in California using data fusion methods, but, to fully understand wildfire-aerosol interactions, research focus should expand beyond PM to examine the broader aerosols' spectrum. Additionally, investigations across other regions of the contiguous US are needed to generalise California-specific findings. Therefore, this paper aims to examine the causal effect of wildfires on harmful aerosol concentrations across the entire country.

 Studying causal relationships is essential because correlation analysis alone cannot guide effective intervention strategies. Policymakers can develop targeted mitigation strategies by identifying which wildfire characteristics directly cause specific air quality deterioration. Recent research has focused on correlation analyses more than causal investigations, with significant attention given to quantifying wildfire contributions to air pollution \citep{burke2023contribution} and examining the associations between wildfires and extreme air quality degradation \citep{knorr2016air}. \citet{jaffe2020wildfire} examined the relationship between wildfires and elevated fine particulate matter exposure, discovering that wildfire smoke experiences atmospheric chemical changes that modify its composition and produces additional pollutants, like ozone (O$_3$) and \mbox{secondary} organic aerosols. \citet{schneider2021air} identified that regions influenced by wildfires exhibit higher concentrations of PM, nitrogen dioxide ($\text{NO}_2$), and CO, but not $\text{O}_3$. \citet{larsen2022spatial} examined how wildfires contribute proportionally to total air pollution, demonstrating that wildfire-pollution relationships exhibit non-linear patterns with significant spatial variation. In addition, \citet{vasilakopoulou2023rapid} and \citet{he2024formation} indicated that wildfire emissions undergo complex physico-chemical transformations that may be influenced by existing atmospheric conditions. This makes it essential to investigate the heterogeneous effects of wildfire smoke across varying air pollution levels, to understand how pre-existing pollution modulates secondary aerosol formation pathways and subsequent health impacts. The collective evidence suggests that wildfire intensity and air quality relationships are most likely nonlinear: incremental increases in fire activity beyond specific air quality levels may lead to disproportionate deterioration in air quality. Moreover, the spatial variation of wildfire impacts on air quality across geographical regions also constitutes a fundamental aspect of our investigation. Hence, we develop a spatial causal inference framework to accurately capture variation in these relationships across different locations and environmental conditions, as well as to assess how the causal relationships change across the air quality distribution.

\subsection{Spatio-Temporal and Causal Models}
Spatio-temporal models are key in climate science. The hierarchical Bayesian framework of \citet{wikle1998hierarchical} has been widely adopted for climate field reconstruction and forecasting, where the complex space-time dependencies are decomposed into conditional probabilistic structures. Since 2007, the widely-used European Centre for Medium-Range Weather Forecasts (ECMWF) reanalysis model has improved physical parameterizations and significantly reduced systematic errors in simulations of tropical precipitation, atmospheric waves, and Northern Hemisphere circulation patterns \citep{jung2010ecmwf}. Conventional spatio-temporal models are formulated for passively observed data rather than actively intervened systems. Yet, addressing causal questions such as ``\textit{What effect would a specific climate attribute have on the air quality index?}'' necessitates either interventional data from, for instance, randomized controlled trials \citep{nitsch2006limits}, or a causal model compatible with observational data. Given the scarcity of interventional datasets, there is a critical need for causal models that explicitly represent the underlying data-generating mechanisms rather than simply capturing their observational patterns.

Causal models quantify cause-and-effect relationships between random variables, for example, by examining how the expected value of $Y$ changes when $X$ intervenes (manipulates) \citep{ding2018causal}. However, causal models are typically designed for independent and identically distributed (IID) or time series data, and while well-established causality theories exist in pharmaceutical and financial studies \citep{hofler2005causal, riegg2008causal}, their application in environmental studies is limited. Thus, causal models frequently use assumptions that may prove untenable across spatial applications. A fundamental prerequisite---that no unmeasured variables simultaneously influence both the exposure and outcome variables (thus functioning as confounders)---is an essential assumption in classical causal analysis, as violations of this condition can introduce significant bias into effect estimation. Spatial confounding presents another particular challenge when unmeasured spatially-structured variables affect both exposure and outcome, potentially inducing spurious associations through shared spatial patterns, and are discussed in \citet{dupont2022spatial+} and \citet{gong2025causal}.

The current state-of-the-art spatial causal modelling framework is the Latent Spatial Confounder Model (LSCM) proposed by \citet{christiansen_toward_2022}, which provides a general framework for estimating the mean causal effect in spatio-temporal data. However, it remains uncertain whether causal mechanisms are preserved under extreme circumstances, with many causal relationships only becoming evident during extreme events, and whether estimation of average causal effects is sufficient for causal discovery \citep{gnecco_causal_2020,jiang2025separation}. For instance, \cite{cox2014multivariate} noted that large interventions often carry information that is likely to be causal, while \cite{diffenbaugh2017quantifying} found that historical warming exerts disproportionately stronger causal effects on extreme temperature events, with significantly larger impacts on the upper quantiles of the temperature distribution. Moreover, the increasing frequency of wildfires due to climate change has highlighted the importance of modelling extreme events \citep[see, e.g.,][]{richards2022regression,koh2023spatiotemporal,majumder2025semi}. We highlight the importance of looking beyond the mean via the following illustrative example.

\subsection{An Illustrative Example of the Limitations of Classical Causal Models}

The assumption of linearity is foundational to classical causal discovery algorithms, including constraint-based methods like PC \citep{kalisch2007estimating} and FCI \citep{spirtes2000causation}, as well as functional models like LiNGAM~\citep{hyvarinen2000independent}. However, linearity can be problematic for extreme events; linear algorithms may fail to detect nonlinearities in the tails, incorrectly suggesting conditional independence (via $d$-separation; see Appendix~\ref{app:bci}). In the following example, we demonstrate that quantile regression reveals nonlinear causal links missed by linear approaches, particularly when the data generating process exhibits heterogeneity across the conditional distribution.


\begin{example}
    \label{ex:wrongly}
     Let \(H \sim \text{Unif}(10, 15)\) be a hidden confounder of an exposure $X$, assumed to be expressed as \({X = H + E_1},\) where \(E_1 \sim \text{Unif}(0.1, 1)\). Then, define an outcome variable \[Y = 0.5Z_1 + 0.5Z_2 + 0.1 E_2,\] where $Z_1 \sim \mathcal{N}\left(-X/2 + H, 0.5\right)$, $Z_2 \sim \Gamma\left((X + H)/20, 0.1\right)$, \(E_2 \sim t_3\), where $t_3$ is the Student's $t$ distribution with three degrees of freedom, and $\Gamma(\alpha, \lambda)$ denotes the Gamma distribution with shape $\alpha > 0$ and rate $\lambda >0$. Theoretically, the conditional mean of $Y$ does not depend on $X$, as ${\E[Y\mid X = x] = 0.5 \times (-0.5 x + H) + 0.5 \times (0.5 x + 0.5 H) = 0.75 H}$. Thus, we would expect that linear (mean) regression models cannot recover the underlying relationship between $Y$ and $X$, even though we know that one exists. Instead, we find useful information in other parts of the conditional distribution, $Y\mid (X,H)$.

    The Directed Acyclic Graph (DAG) defining the model for Example \ref{ex:wrongly} has edges $X \to Y$, $H \to Y, H \to X$, and illustrates the causal relationship between the exposure ($X$), the response ($Y$), and hidden confounder ($H$). In this model, $H$ serves as a common cause that influences both $X$ and $Y$, thereby introducing confounding bias in the estimation of the causal effect of $X$ on $Y$. 
    In Example \ref{ex:wrongly}, the hidden confounder $H$ makes causal inference difficult by masking the true effect of $X$ on $Y$. Traditional methods' limitations necessitate novel approaches that account for both hidden confounding and non-linearity. To illustrate this, we generate $n = 50{,}000$ samples under Example \ref{ex:wrongly} and use the LiNGAM and PC causal discovery algorithms to identify the causal link from $X$ to $Y$; for simplicity, we assume that we know the hidden confounder $H$. However, both classical algorithms fail to identify the correct causal relationship. Specifically, LiNGAM estimates no causal link from $X$ to $Y$, and the PC algorithm returns no edges between nodes $X$ and $Y$ in the estimated causal graph. PC algorithm with parametric or nonparametric conditional independence tests both fail to detect the $X \to Y$ relationship, as the causal effect operates through distributional shape rather than conditional mean; for details about causal inference preliminaries, see Appendix~\ref{app:bci}. 
   
   On the other hand, if we know that there is a causal link from $X$ to $Y$, we can estimate its causal strength using regression methods. With linear (mean) regression, where the estimand is the conditional expectation, the ordinary least squares estimated regression coefficient for $X$ is $0.0089$ with $p$-value $0.683$. In contrast, using linear quantile regression at quantile-level $\tau = 0.9$ (i.e., $90\%$ quantile), where the estimand is the conditional $\tau$-quantile \citep[see, e.g.,][]{koenker2005quantile}, the estimated coefficient for $X$ is $0.2518$ with $p$-value $5\times 10^{-5}$. While the estimated linear (mean) regression coefficient lacks statistical significance, the estimated quantile regression coefficient exhibits strong significance. In Example~\ref{ex:wrongly}, traditional expectation-based causal frameworks fail to provide sufficient evidence for causal relationships, and traditional conditional independence tests also fail to recognize the true underlying causal relationships. Therefore, a more general and robust framework—via conditional quantile estimation—provides crucial information for accurate and confident causal identification. Additionally, conditional quantile approaches relax the assumption of IID errors required by linear models; see Section~\ref{sec:qlscm}.
\end{example}
In this paper, we propose the Quantile-based Latent Spatial Confounder Model (QLSCM) framework to realize several objectives: 1) account for extreme events using a quantile-based causal estimand, 2) identify causal links in spatio-temporal processes when linear regression fails, and 3) relax the IID errors assumption of traditional expectation-based models. For our proposed framework, we prove consistent estimation of the proposed causal estimand, and use an illustrative simulation study to provide empirical evidence of the efficacy of the QLSCM. To examine causal relationships between US wildfire intensity and aerosol-based air quality measurements, we apply our model to US climate observations spanning 2003 to 2020.
\end{hyphenrules}

\subsection{Structure of the Paper}
\begin{hyphenrules}{nohyphenation}
In Section~\ref{sec:data}, we describe the dataset utilized in our data analysis. Section~\ref{sec:background} presents the essential mathematical preliminaries. Section~\ref{sec:qlscm} introduces our proposed QLSCM framework and provides theoretical support for its use. Section~\ref{sec:5} employs simulation studies to demonstrate the robustness and necessity of using conditional quantiles rather than conditional expectations, and presents a representative example to illustrate the QLSCM framework. Section~\ref{sec:6} provides our data analysis of the causal relationship between wildfire intensity and air quality indicators in the US. Section~\ref{sec:8} concludes with potential avenues for future research.
\end{hyphenrules}

\section{US Wildfire and Air Quality Data}
\label{sec:data}

\begin{hyphenrules}{nohyphenation}
Our wildfire and air quality data are defined on a daily temporal scale and are aggregated over spatial grid cells. The observation period spans 2003--2020, focusing only on the peak wildfire season of June to November. The outcome data $Y_{\bs s}^t$ comprise aerosol optical depth (AOD; nm) observations aggregated over a spatial grid cell $\bs s \in \mathbb{R}^2$ and time period $t \in \mathbb{N}$. Spatial locations are arranged on a regular latitude/longitude grid covering the contiguous United States (CONUS) with spatial resolution $0.25^\circ \times 0.25^\circ$; see Figure~\ref{fig:visuala}, which provides the monthly mean observed AOD for September 2004. Observations of AOD are obtained from NASA's Moderate Resolution Imaging Spectroradiometer (MODIS) Collection 6.1 \citep{AOD}. While MODIS reports AOD at various spectral bands, we exclusively consider the $550$ nanometer band, which is the most widely used in scientific research and applications \citep{gupta2020high}. AOD quantifies the presence of many aerosols (including urban haze, smoke particles, desert dust, and sea salt) distributed throughout an atmospheric column extending from the Earth's surface to the atmosphere's upper boundary. Specifically, aerosol particles in the atmosphere impede sunlight by absorption or scattering, and AOD indicates the degree to which these particles prevent direct sunlight from reaching the ground. The AOD measurement encompasses aerosols present in the air while excluding water vapor, Rayleigh scattering, and various wavelength-dependent trace gases that are extensively studied~\citep{zielinski2024impact}, such as $\text{O}_3$, $\text{NO}_2$, $\text{CO}_2$, and $\text{CH}_4$; by designating this measurement as our air quality response variable, we eliminate the need to examine causal relationships between wildfires and the aforementioned widely-studied gases.
\end{hyphenrules}

\begin{figure}[t!]
    \centering
    \begin{subfigure}[b]{0.425\textwidth}
        \centering
        \includegraphics[width=\textwidth]{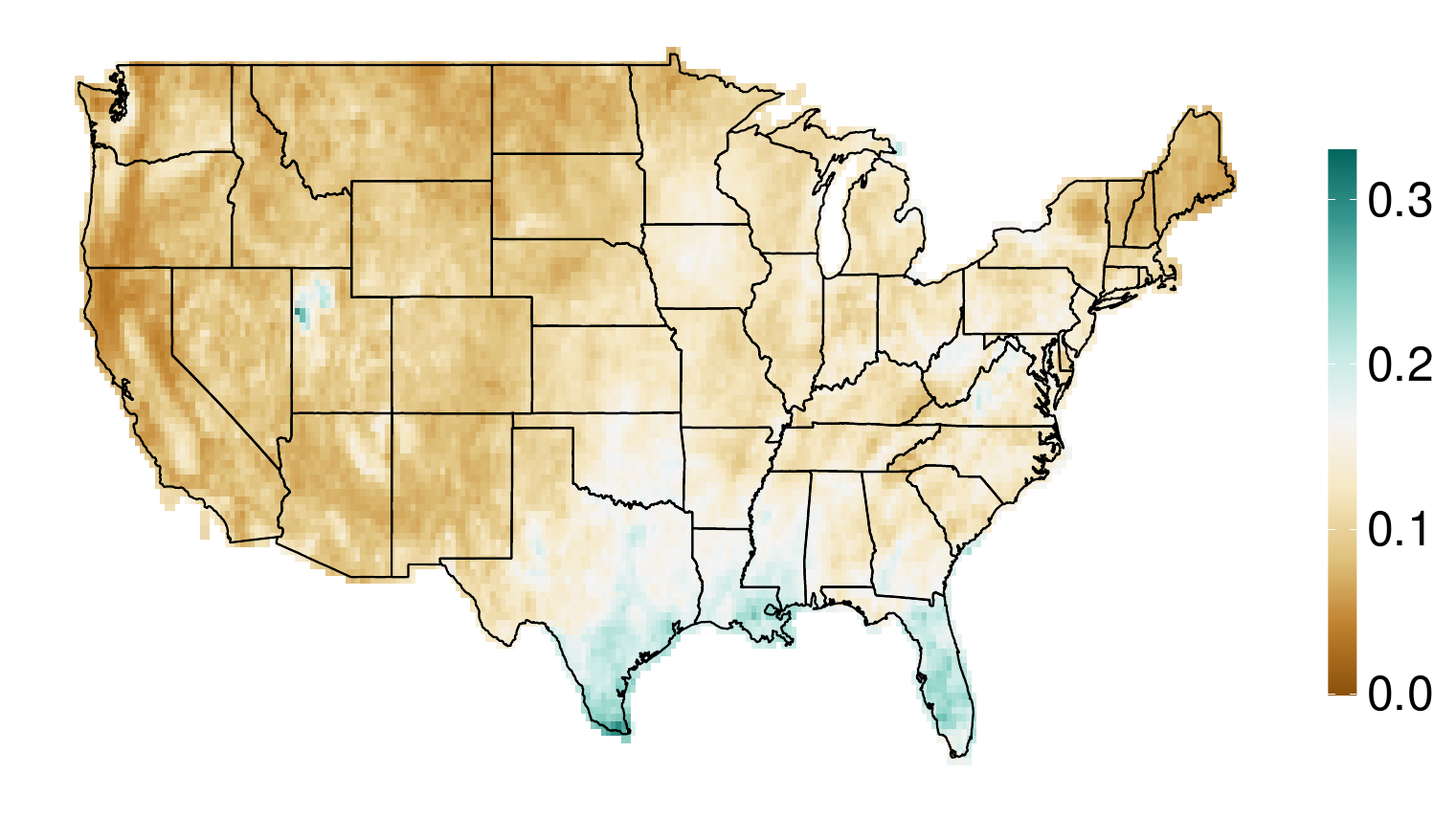}
        \caption{\textbf{Aerosol optical depth}}
        \label{fig:visuala}
    \end{subfigure}
    \hfill
    \begin{subfigure}[b]{0.425\textwidth}
        \centering
        \includegraphics[width=\textwidth]{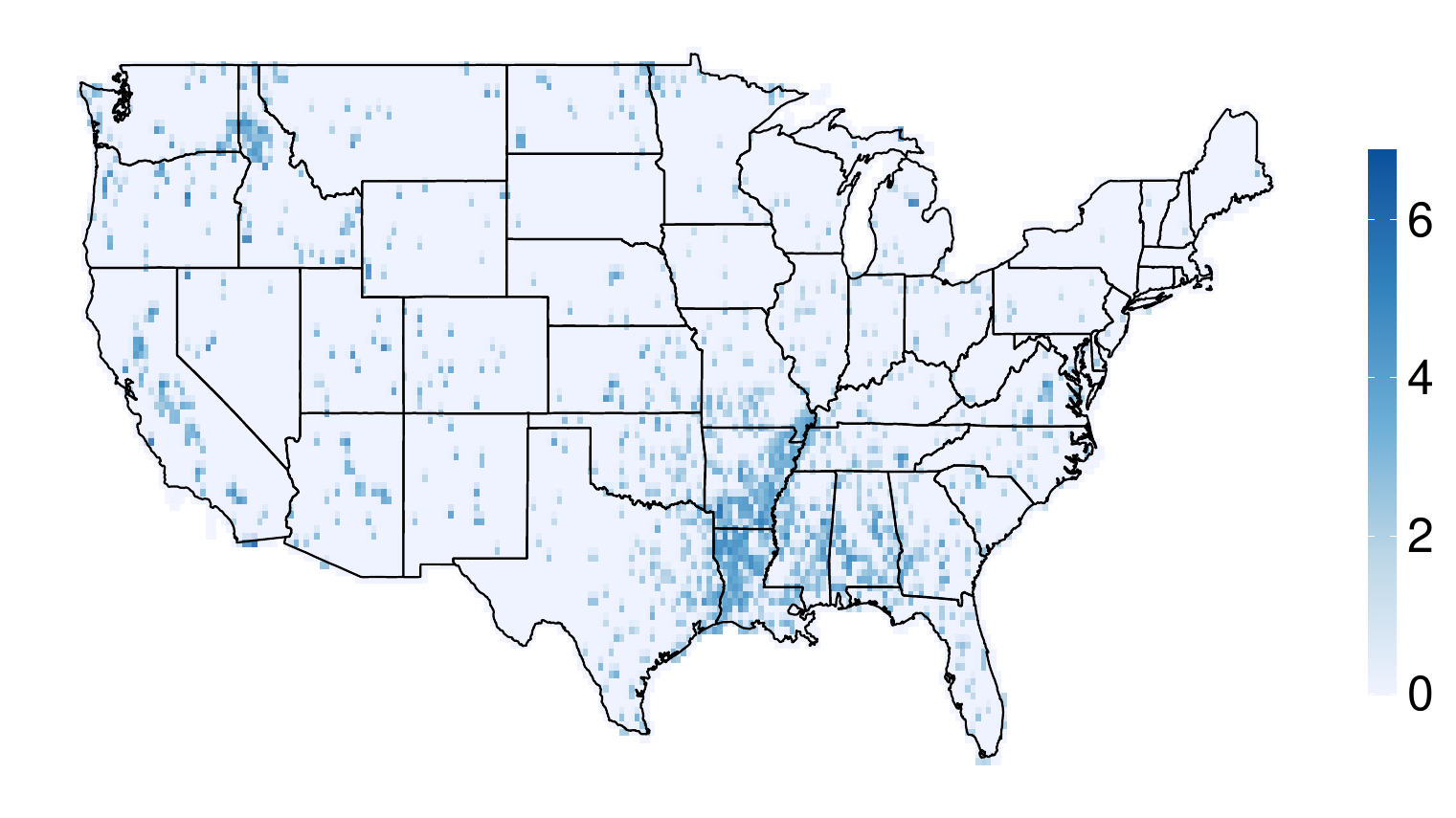}
        \caption{\textbf{Fire radiative power}}
        \label{fig:visualb}
    \end{subfigure}
    
    \vspace{0.3cm}
    
    \begin{subfigure}[b]{0.425\textwidth}
        \centering
        \includegraphics[width=\textwidth]{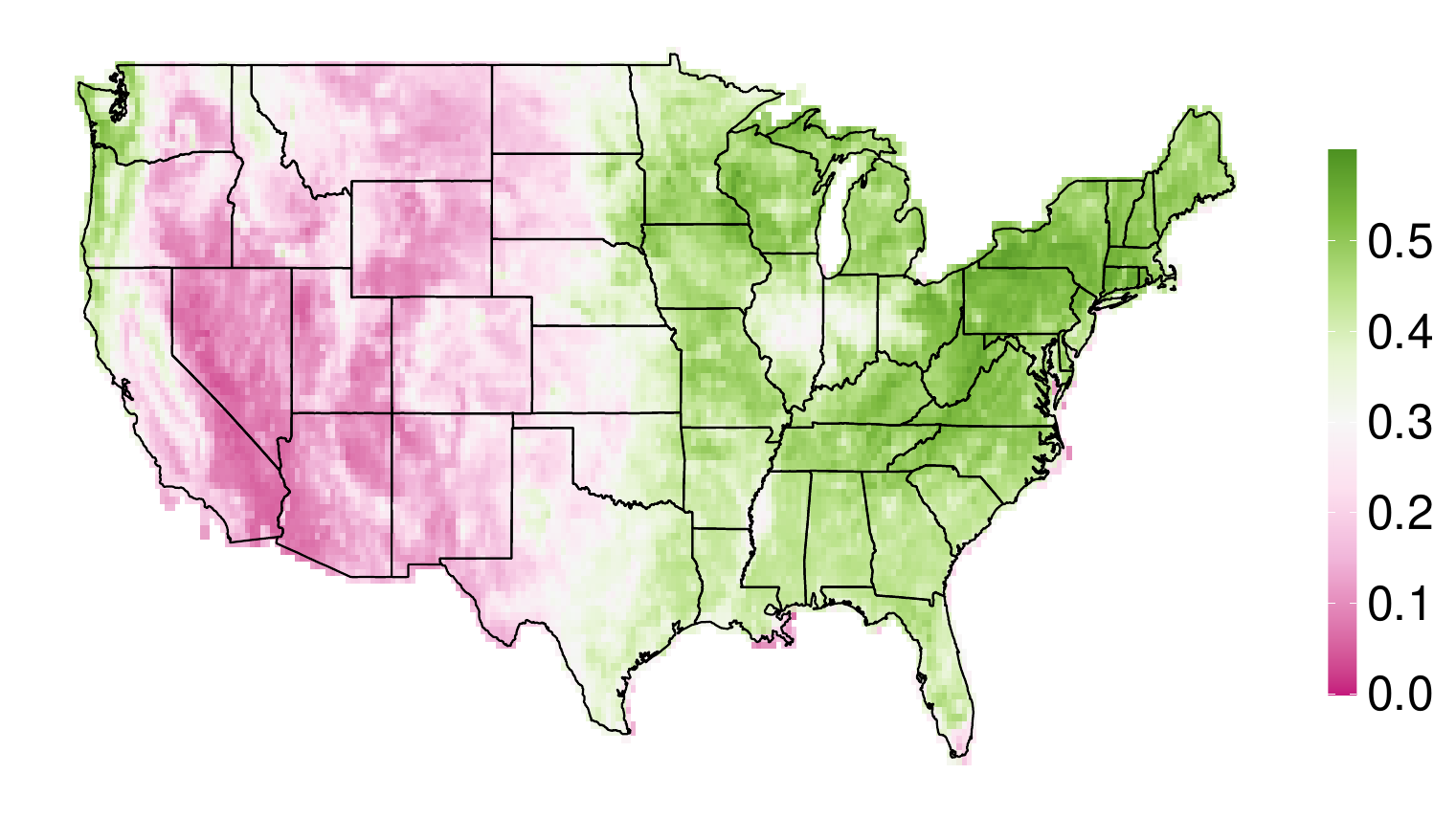}
        \caption{\textbf{Enhanced vegetation index}}
        \label{fig:visualc}
    \end{subfigure}
    \hfill
    \begin{subfigure}[b]{0.425\textwidth}
        \centering
        \includegraphics[width=\textwidth]{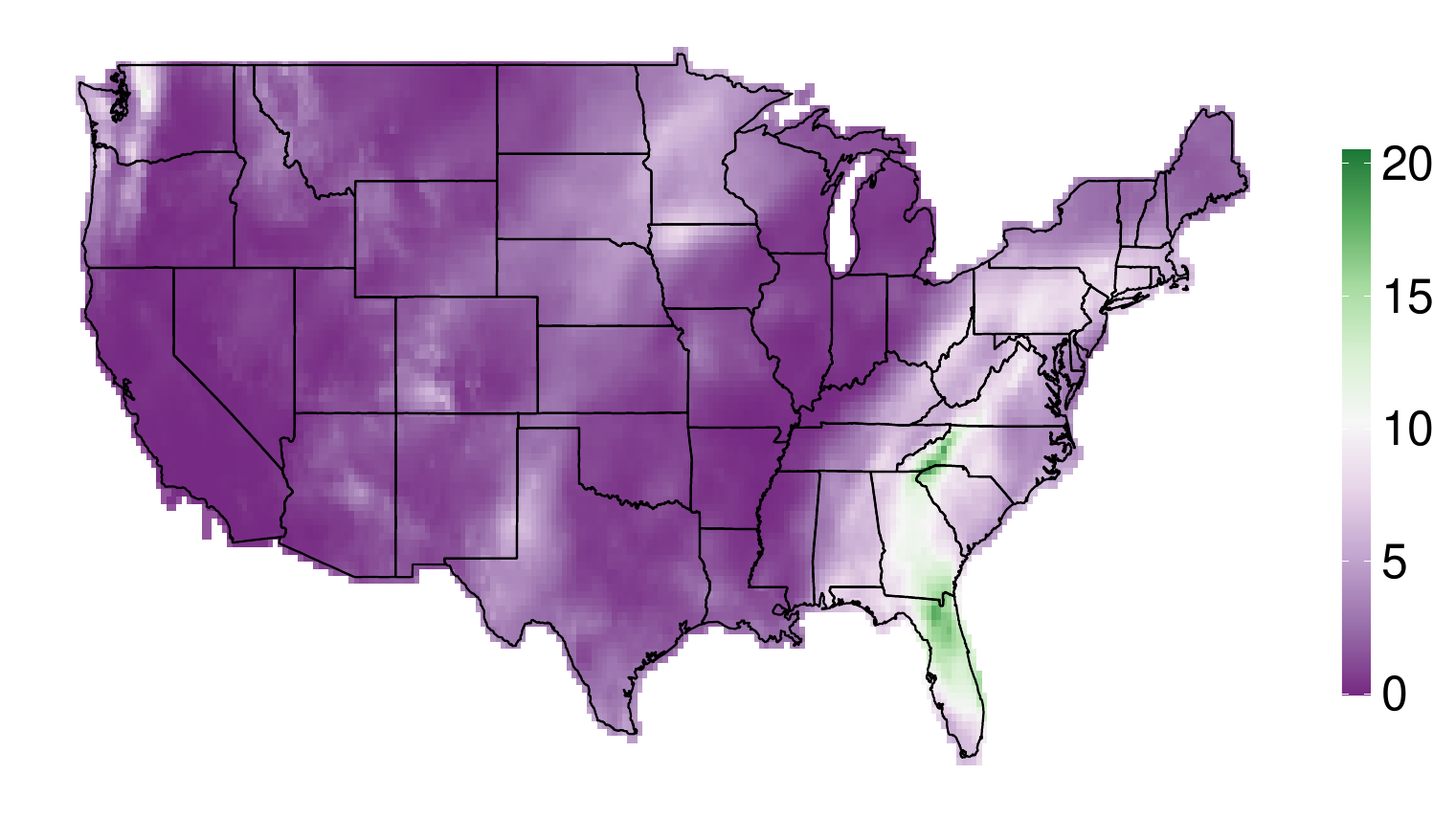}
        \caption{\textbf{Precipitation}}
        \label{fig:visuald}
    \end{subfigure}
    
    \vspace{0.3cm}
    
    \begin{subfigure}[b]{0.425\textwidth}
        \centering
        \includegraphics[width=\textwidth]{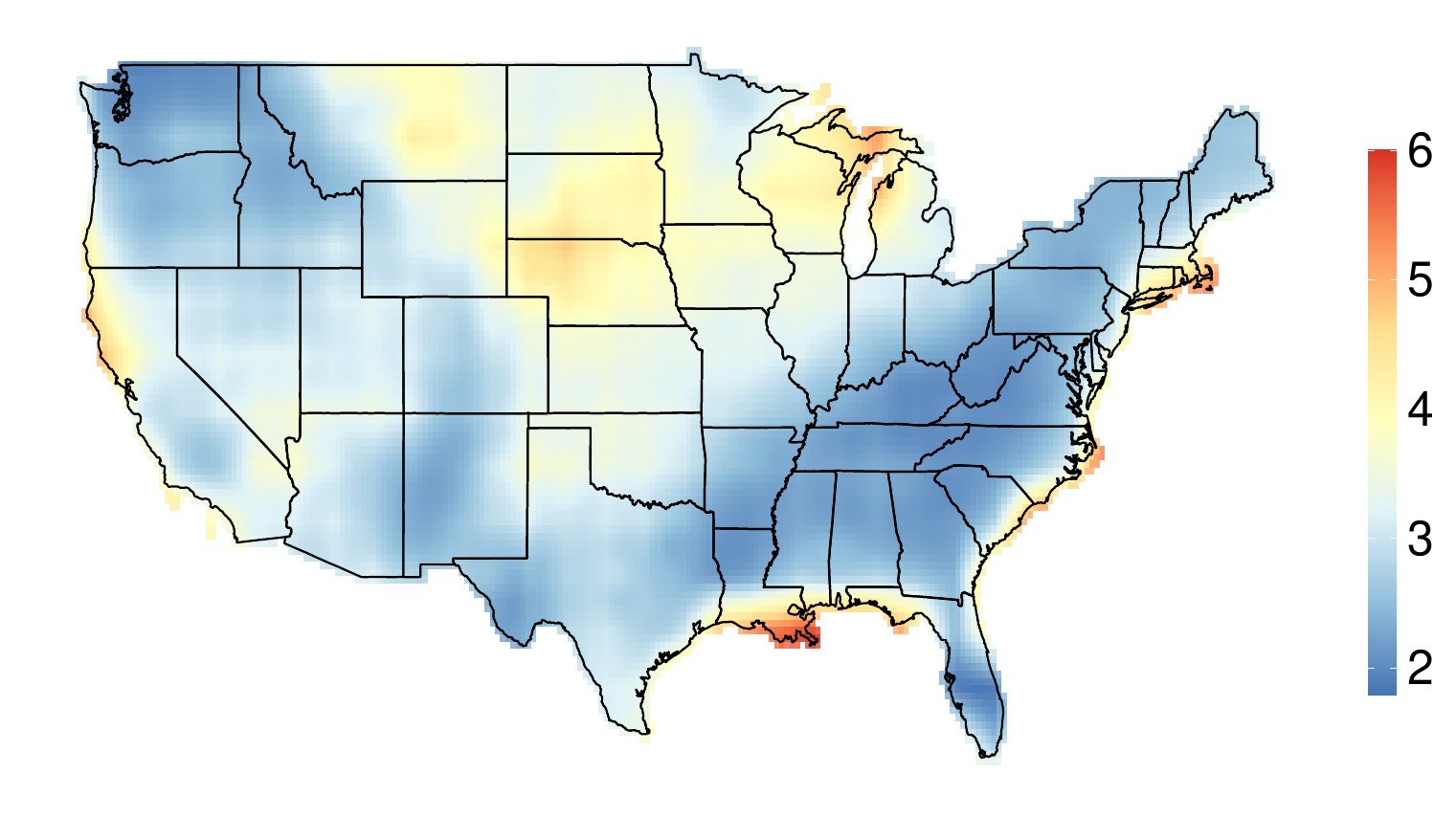}
        \caption{\textbf{Wind speed}}
        \label{fig:visuale}
    \end{subfigure}
    \hfill
    \begin{subfigure}[b]{0.425\textwidth}
        \centering
        \includegraphics[width=\textwidth]{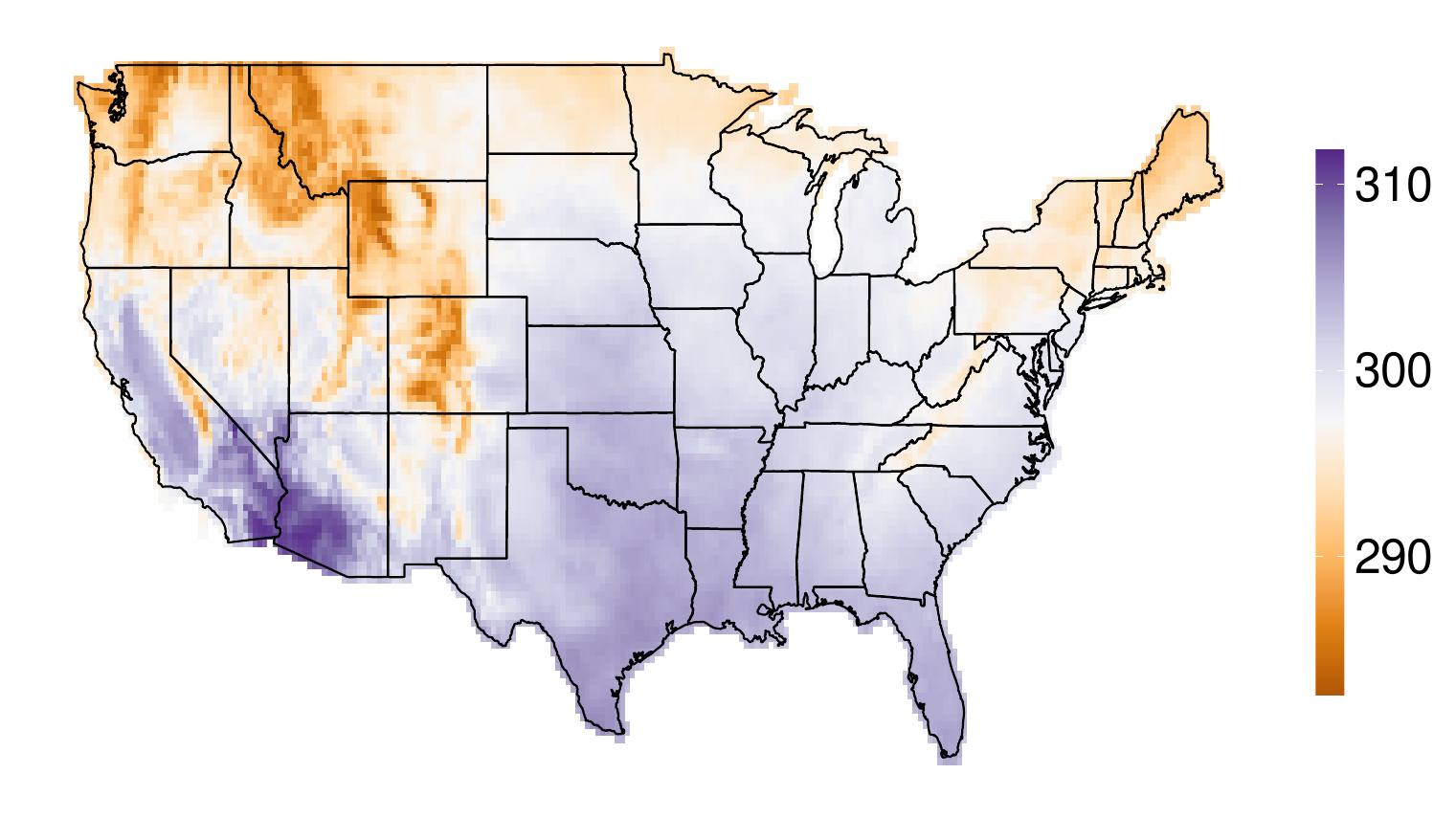}
        \caption{\textbf{Maximum temperature}}
        \label{fig:visualf}
    \end{subfigure}
    
    \caption{Maps of (a) monthly mean of $Y_{\bs s}^t$ [AOD; $\text{nm}$], (b) monthly mean of $X_{\bs s}^t$ [FRP; mW], and 4 observed confounders ($\mbW_{\bs s}^t$): (c) monthly mean of enhanced vegetation index [EVI; unitless], (d) monthly mean of precipitation [PRCP; mm], (e) monthly mean of near-surface wind speed [WS; m/s], (f) monthly mean of daily-maximum temperature [TMAX; K] for September 2004.}
    \label{fig:visual_all}
\end{figure}
\begin{hyphenrules}{nohyphenation}
For the exposure $X_{\bs s}^t$, we select fire radiative power (FRP; mW) as our measure of wildfire intensity. FRP captures the pixel-integrated fire radiative power for a $1$ km $\times$ $1$ km grid box. The MODIS Fire and Thermal Anomalies algorithm \citep{giglio2003enhanced} identifies these cells as containing at least one active fire, and then the values are aggregated to the same $0.25^\circ \times 0.25^\circ$ spatial grid as the AOD observations; see Figure~\ref{fig:visualb}, which provides monthly mean observed FRP for September 2004. While wildfire impacts can be observed through metrics such as burnt area \citep[see, e.g., ][]{richards2022regression}, these datasets present limitations regarding spatial resolution and incomplete coverage across certain regions of the CONUS. We note that FRP measurements may include both wildfires and intentional landscape fires, such as prescribed burns. However, our analysis proceeds under the assumption that FRP values predominantly reflect wildfire activity, as intentional landscape fires typically generate significantly lower intensity compared to wildfires \citep[see, e.g.,][]{fernandes2003review,ryan2013prescribed}.

While some confounders are unobservable, we include in our analysis a vector of observed confounders, $\mb W_{\bs s}^t \in \mathbb{R}^k$. These include $k=4$ climate variables known to confound wildfire-air quality relationships \citep{zhang2023wildfire}: Enhanced Vegetation Index (EVI; unitless) from MODIS Collection 6.1, Precipitation (PRCP; mm) and Maximum Temperature (TMAX; K) from the Gridded Surface Meteorological Dataset (GridMET) \citep{abatzoglou2013development}, and Near-Surface Wind Speed (WS; m/s) integrated from two data sources: 1) gridded historical observations spanning 2003--2015, and 2) CESM2 (Community Earth System Model version 2) predictions \citep{danabasoglu2020community} covering 2016--2020. All variables are aggregated to the same regular spatio-temporal grid; see Figures~\ref{fig:visualc}--\ref{fig:visualf} for observed monthly means for September 2004.
\end{hyphenrules}

\section{Preliminaries}
\label{sec:background}

\subsection{Background on Spatio-Temporal Processes and Notation}
\begin{hyphenrules}{nohyphenation}
Following \citet{christiansen_toward_2022}, we consider a $p$-dimensional spatio-temporal process $\mathbf{Z}$ defined over spatial locations $\bs s \in \mathbb{R}^2$ and discrete time points $t \in \mathbb{N}$. Each realization $\mb z$ of $\mathbf{Z}$ is a function $\mb z: \mathbb{R}^2 \times \mathbb{N} \to \mathbb{R}^p$ such that integrals of the form $\int_A \mbz(\bs s, t) \, \mathrm{d}\bs s$ exist for any $A \subset \mathbb{R}^2$ and any $t \in \mathbb{N}$. We further assume that the density function (when it exists) of $\mbZ_{\bs s}^t$, the marginalized $p$-dimensional random vector at location $\bs s$ and time $t$, is measurable. 

We adopt the following notation: $\mathbf{Z}_{\bs s}$ represents the time series $(\mbZ_{\bs s}^t)_{t\in\mathbb{N}}$ at location $\bs s \in \mathbb{R}^2$; $\mathbf{Z}^t$ denotes the spatial process $(\mbZ_{\bs s}^t)_{\bs s\in\mathbb{R}^2}$ at time $t$; and $\mathbf{Z}^{(I)}$ refers to the marginal coordinate process consisting of coordinates indexed by $I \subseteq \{1, \dots, p\}$. We say that $\mathbf{Z}$ is \textit{weak-sense stationary} if the distribution of $\mbZ_{\bs s}^t$ is the same for all $(\bs s, t) \in \mathbb{R}^2 \times \mathbb{N}$, and \textit{time-invariant} if the spatial process remains constant over time, i.e., $\mathbf{Z}^1 = \mathbf{Z}^2 = \cdots$ almost surely.
\end{hyphenrules}

\subsection{Causal Graphical Models for Spatio-Temporal Processes}
\label{sec:causal_graph}

\begin{hyphenrules}{nohyphenation}
The components of a multivariate spatio-temporal process $\mathbf{Z}$ can be partitioned into disjoint groups. Our analysis aims to identify causal relationships \emph{between} different groups of components, represented as $\mathbf{Z}^{(I)}$, without specifying the causal structure \emph{within} each group. 
The Latent Spatial Confounder Model (LSCM) of \citet{christiansen_toward_2022} was introduced for this purpose: it decomposes the joint distribution of $\mathbf{Z}$ into \mbox{disjoint} groups, where the distribution of each group is specified conditionally on other groups. This decomposition naturally induces a graph that encodes dependencies between groups. By specifying how the distribution changes under interventions, we can interpret these dependencies causally. For formal background on graphical models, see Appendix~\ref{app:bci}. Full details on the LSCM are provided in Appendix~\ref{app:lscm_defn}.

A \textit{causal graphical model} for the $p$-dimensional spatio-temporal process $\mathbf{Z}$ consists of three components: 1) a partition $\mathcal{I} = \{I_1, \dots, I_h\}$ that divides the $p$ components into $h$ disjoint groups; 2) a directed acyclic graph (DAG) $\mathcal{G}$ with nodes corresponding to these groups $I_1, \dots, I_h$; and 3) a family $\mathcal{P}$ of conditional distributions that specifies, for each group $I_j$, the distribution of $\mathbf{Z}^{(I_j)}$ given its parent groups. For group $I_j$, we denote its parents by $\text{PA}_j$, which consists of all groups $I_i$ with a directed edge $I_i \to I_j$ in $\mathcal{G}$. When $I_j$ has no parents, its distribution is unconditional. 
In obvious notation, the graph $\mathcal{G}$ and conditional distributions $\mathcal{P}$ together uniquely determine the joint distribution $\mathbb{P}$ of $\mathbf{Z}$ through the factorization 
    $\mathbb{P}(\mathbf{Z}^{(I_1)}, \dots, \mathbf{Z}^{(I_h)}) = \prod_{j=1}^h \mathbb{P}(\mathbf{Z}^{(I_j)} \mid \mathbf{Z}^{(\text{PA}_j)}),$
where we order the groups such that parents always precede their children. We call $\mathbb{P}$ the \textit{observational distribution}. 
The causal interpretation arises from considering interventions. An \textit{intervention on group $I_j$} replaces the conditional distribution $\mathbb{P}(\mathbf{Z}^{(I_j)} \mid \mathbf{Z}^{(\text{PA}_j)})$ with a new distribution $\tilde{\mathbb{P}}(\mathbf{Z}^{(I_j)} \mid \mathbf{Z}^{(\text{PA}_j)})$, while leaving all other conditional distributions unchanged. This replacement induces a new joint distribution $\tilde{\mathbb{P}}$ via the above factorization, 
which we call the \textit{interventional distribution}. Crucially, the intervention affects only how $\mathbf{Z}^{(I_j)}$ depends on its parents, without altering the conditional distributions of other groups. This modularity property, analogous to that in structural causal models \citep{gnecco_causal_2020}, justifies interpreting the edges in $\mathcal{G}$ as representing causal relationships. We refer to $\mathcal{G}$ as the \textit{causal structure of \,$\mathbf{Z}$}, and write $[\mb Z] = [\mb Z^{(I_h)} \mid \mbZ^{\left(\text{PA}_h\right)}] \cdots [\mb Z^{(I_1)}]$ to denote this structure.
\end{hyphenrules}

\section{Quantile-based Latent Spatial Confounder Model}
\label{sec:qlscm}

\begin{hyphenrules}{nohyphenation}
In this section, we introduce our novel quantile-based latent spatial confounder model (QLSCM). This approach modifies the original latent spatial confounder model (LSCM) (for details, see Appendix~\ref{app:lscm_defn}) by considering conditional quantiles, instead of conditional expectations, as the causal estimand. We also introduce observed \mbox{confounders}, \mbox{$\mbW \in \mathbb{R}^k$}, and relax independently and identically distributed (IID) random error assumption of the LSCM. When the causal estimand \mbox{depends} on the quantile level $\tau \in (0, 1)$, we establish the terminology of $\tau$-quantile causal effects when referring to any causal estimands that depend on the quantile level $\tau$.

In the following Sections~\ref{subsec:qlscm} and \ref{subsec:qreg}, we define the QLSCM and introduce spatio-temporal quantile regression, using the causal graphical models discussed in Section~\ref{sec:causal_graph}. In Section~\ref{subsec:target}, we establish the causal estimand that we aim to measure (the target parameter) and demonstrate that this estimand can be causally interpreted and identified from observable distributions. In Section~\ref{subsec:estimation}, we develop a consistent estimator for the true underlying causal estimand. 
\end{hyphenrules}

\subsection{Definition of the QLSCM} \label{subsec:qlscm}
\begin{definition}[QLSCM] \label{def:qlscm}
    \sloppy{Consider a multivariate spatio-temporal process $(\mb X, \mb Y, \mb H, \mb W) = $ $(\mbX_{\bs s}^t, Y_{\bs s}^t, \mbH_{\bs s}^t, \mbW_{\bs s}^t)_{(\bs s,t) \in \mathbb{R}^2 \times \mathbb{N}}$} over response $Y^t_{\bs s}\in \mathbb{R}$, exposure $\mbX_{\bs s}^t \in \mathbb{R}^{d}$, hidden \mbox{confounder} $\mbH_{\bs s}^t \in \mathbb{R}^\ell$, and observed confounder $\mbW_{\bs s}^t \in \mathbb{R}^k$. We call a causal graphical model over $(\mb X, \mb Y, \mb H, \mb W)$, with causal structure $[\mb Y \mid \mb X, \mb H, \mb W][\mb X \mid \mb H, \mb W][\mb H, \mb W]$, a QLSCM if the following conditions hold for the observational distribution: 1) $\mb H$ is time-invariant, weak-sense stationary, and spatially ergodic; and 2) there exists a \mbox{measurable} function $f : \mathbb{R}^{d+\ell+k+1} \rightarrow \mathbb{R}$ and an 
identically distributed sequence $\bs\eps^1, \bs\eps^2, \dots$ of weak-sense stationary spatial error processes, independent of $(\mb X, \mb H, \mb W)$, such that
    \begin{equation}
            Y_{\bs s}^t := 
            f(\mbX_{\bs s}^t, \mbH_{\bs s}^t, \mbW_{\bs s}^t, \eps_{\bs s}^t) \quad \text{for all } (\bs s, t) \in \mathbb{R}^2 \times \mathbb{N}.
            \label{eq:model}
    \end{equation} Here, $f(\cdot)$ represents the true, unknown model, and $\eps_{\bs s}^t$ is the error term.
\end{definition}

\begin{remark}
Hereafter, we assume that $(\mb X, \mb Y, \mb H, \mb W)$ comes from a QLSCM. This framework relies on a modified LSCM and assumes that, for every $(\bs s, t)$, i) $Y_{\bs s}^t$ depends on $(\mb X, \mb H, \mb W)$ only via $(\mbX_{\bs s}^t, \mbH_{\bs s}^t, \mbW_{\bs s}^t)$, and ii) that this dependency is sufficient to reveal causal relationships and remains the same for all points in the space-time domain. We additionally introduce observed confounders $\mbW$, distinguishing them from the exposure $\mbX$. In any inferential procedures involving $\mbX$, the vector $\mbW$ plays an equivalent role. It is well known that controlling for observed confounders will reduce estimation bias~\citep{pearl_causality_2009}. For instance, when fitting a quantile regression model (discussed subsequently), $\mbW$ functions as a covariate alongside other predictors. 

Within the QLSCM framework, quantile regression serves as the tool for identifying causal relationships. This approach places no constraints on how $\mbX$ is marginally distributed and leverages the robustness of quantile regression to handle heavy-tailed data when modeling causal influences. Thus, we do not require the IID errors assumption of \citet{christiansen_toward_2022} since the classical theory of quantile regression does not need this prerequisite~\citep{koenker2005quantile}. The QLSCM can be further generalized to include observed lagged covariates and confounders; this is a natural extension.
\end{remark}

\subsection{Spatio-Temporal Quantile Regression} \label{subsec:qreg}
Our QLSCM framework is built upon quantile regression adapted to the spatio-temporal context. We consider a setting where $Y_{\bs s}^t \in \mathbb{R}$ represents a continuous random variable at spatial location $\bs s \in \mathbb{R}^2$ and time instance $t \in \mathbb{N}$, with a distribution function $F_{Y_{\bs s}^t}(y) := \PP(Y_{\bs s}^t \le y)$. For any probability level $\tau \in (0,1)$, we define the $\tau$-quantile of $Y_{\bs s}^t$ as $Q_{\tau}(Y_{\bs s}^t) := F_{Y_{\bs s}^t}^{-1}(\tau) = {\inf\{y : F_{Y_{\bs s}^t}(y) \ge \tau\}}$. The conditional $\tau$-quantile of $Y_{\bs s}^t$ for observed values $\mathbf{X}_{\bs s}^t = \mathbf{x}$ is defined as
\[
Q_{\tau}(Y_{\bs s}^t \mid \mathbf{X}_{\bs s}^t = \mathbf{x}) := F_{Y_{\bs s}^t \mid \mathbf{X}_{\bs s}^t}^{-1}(\tau \mid \mathbf{x}) = \inf\{y : F_{Y_{\bs s}^t \mid \mathbf{X}_{\bs s}^t}(y \mid \mathbf{x}) \ge \tau\},
\] 
where  $F_{Y_{\bs s}^t \mid \mathbf{X}_{\bs s}^t}(y \mid \mathbf{x})$ is the conditional distribution function of $Y_{\bs s}^t$ given $\mbX_{\bs s}^t=\mathbf{x}$. Following the logic of spatio-temporal quantile regression as defined in \citet{reich2012spatiotemporal}, our goal is to estimate the conditional $\tau$-quantile function $Q_{\tau}(Y_{\bs s}^t \mid \mathbf{X}_{\bs s}^t)$ via a parametric model $\kappa_{\tau}(\mbX_{\bs s}^t, \bs\beta(\bs s, t, \tau))$, where regression coefficients $\bs\beta(\bs s, t, \tau)$ vary across spatial location $\bs s$, time instance $t$, and quantile level $\tau$. However, courtesy of the time-invariance assumption of the QLSCM, we can regard time instances as replicates of the process at each location (with the \emph{same} realized hidden confounder) and perform site-wise quantile regression; we thus drop the notation $t$ from $\bs\beta(\bs s, t, \tau)$ thereafter. 

Site-wise quantile regression leverages temporal replicates to estimate the conditional distribution of $Y_{\bs s}^t$ given the continuous-valued predictor vector $\mathbf{X}_{\bs s}^t = (X_{1,\bs s}^t, \dots, X_{d,\bs s}^t)^\top$ at location $\bs s$ across multiple time points $t$~\citep[see, e.g.,][]{papadogeorgou_causal_2022}. It optimizes over a parametric function $\kappa_{\tau}(\mathbf{X}_{\bs s}^t, \boldsymbol{\beta}(\bs s, \tau))$ with temporal replicates $t \in \mathcal{T}$ by minimizing the empirical loss
\begin{equation*}
    \min_{\boldsymbol{\beta}(\bs s, \tau)} \sum_{t \in \mathcal{T}} \rho_{\tau}(Y_{\bs s}^t - \kappa_{\tau}(\mathbf{X}_{\bs s}^t, \boldsymbol{\beta}(\bs s, \tau))),
\label{eq:spatial-qr}
\end{equation*}
where $\rho_{\tau}(u) := u(\tau - \mathbb{I}_{\{u < 0\}})$ is the pinball loss function~\citep{koenker2005quantile} and $\boldsymbol{\beta}(\bs s, \tau)$ is the vector of regression coefficients. 
The temporal dimension provides multiple observations for estimating these site-specific coefficients. While linear models where $\kappa_{\tau}(\mathbf{X}_{\bs s}^t, \boldsymbol{\beta}(\bs s, \tau)) := {(\mathbf{X}_{\bs s}^t)}^\top \boldsymbol{\beta}(\bs s, \tau)$ are commonly used, semi-parametric forms for $\kappa_{\tau}$ can also be employed, such as B-splines~\citep{sherwood2022additive}. For simplicity, we adopt the former in our analysis.

\subsection{Inferential Target}
\label{subsec:target}
We specify our inferential target, i.e., the causal estimand, as the spatially-averaged causal effect. 
The QLSCM framework avoids imposing any predetermined parametric form on the underlying spatio-temporal process. In our setting, the underlying spatio-temporal characteristics are only observable through a single instance of the process. 
Thus, it makes more sense to define our causal estimand in relation to this actual observed realization, rather than referring to a theoretical distribution of alternative scenarios that will never actually occur. Following this reasoning, we establish the spatial $\tau$-quantile causal partial effect, and show that estimating this effect is mathematically equivalent to estimating the marginal causal effect using quantile regression methods.

\begin{definition}[Spatial $\tau$-Quantile Causal Partial Effect]
    \label{def:qpe}
    Assume that, at each location $\bs s$ and time $t$, we have a single realization $\mbH_{\bs s}^t = \mb h_{\bs s}^t$ of the hidden confounder vector, and consider a spatial region $S \subseteq \mathbb{R}^2$ with area $\vert S\vert$. Since $\mbH_{\bs s}^t$ is time-invariant, the \textit{spatial $\tau$-quantile causal partial effect} of $\mb X$ on $\mb Y$ can be defined as the function $g^{\tau}_{\textrm{QPE}(\mb X \to \mb Y)}: \mathbb{R}^{d+k} \to \mathbb{R}^{d}$, for every $\mb x \in \mathbb{R}^{d}$ and $\mb w \in \mathbb{R}^k$, given by
        \begin{align}
             g^{\tau}_{\textrm{QPE}(\mb X \to \mb Y)}(\mb x, \mbw) &{:=} 
             \frac{1}{\vert S \vert}\int_{S}\frac{\partial Q_{\tau}(\mb Y \mid \mb X = \mb x, \mb H = \mb h_{\bs s}^1, \mb W = \mb w)}{\partial \mb x} \mathrm{d}\bs s \label{eq:def} \\ 
            &= \frac{1}{\vert S\vert}\int_{S}\frac{\partial Q_{\tau}(f(\mb x, \mb h_{\bs s}^1, \mb w, \eps_0^1))}{\partial \mb x}\mathrm{d}\bs s. \nonumber
    \end{align}
\end{definition}

\begin{remark}
    The partial derivative of the conditional quantile function with respect to $\mb x$ is well-defined since we only focus on the causal effect of $\mb X$ on $\mb Y$ with all other variables unchanged. The observed confounding variables $\mbw$ must be controlled to isolate the unbiased causal effect of $\mbX$ on $\mbY$. When working with linear models, the partial derivative is the regression slope coefficient $\bs\beta(\bs s, \tau)$; we focus on this coefficient as it represents the causal partial effect, and disregard the intercept term. In other words, we are more interested in quantifying how changes in $\mb X$ influence changes in $\mb Y$ (the causal pathway) than in predicting $\mb Y$ itself. We thus take the partial derivative as suggested by \citet{wooldridge2004estimating}.
    
    Consider the surface integral in \eqref{eq:def}; within each infinitesimal area, 
    $\mbh_{\bs s}^1$ denotes a realization of the hidden confounder at location $\bs s$ and time $t=1$. Since we assume $\mbH$ to be time-invariant, any realization $\mbh_{\bs s}^t$ at location $\bs s$ remains constant across all time points $t$ almost surely. This means that observing a realization $\mbh_{\bs s}^1$ at time $t=1$ provides the same information as observing $\mbh_{\bs s}^t$ at any other time $t$. Consequently, we can treat $Q_{\tau}(\mb Y \mid \mb X = \mb x, \mb H = \mbh_{\bs s}^1, \mb W = \mb w)$ as a function of only the spatial location $\bs s$ and quantile level $\tau$.
    Under spatial ergodicity of $\mbH$, our causal estimand converges to the population-level spatial marginal effect as $|S| \to \infty$. We use the sample-level estimand in our analysis, but note that the true population target is 
$\lim_{|S| \to \infty}g^{\tau}_{\textrm{QPE}(\mb X \to \mb Y)}(\mb x, \mbw).$
    
    Following \citet{christiansen_toward_2022}, the time-invariance assumption means we only have a single realization of the spatial process $\mb H^1$. It is therefore more practical to base our inference on the observed data rather than on hypothetical, unobservable outcomes. This reasoning supports treating the hidden confounder and observed confounder as a constant for any given spatial location $\bs s \in \mathbb{R}^2$, anchoring the analysis in empirical reality.
\end{remark}



Causal interpretation of our proposed spatial $\tau$-quantile partial effect (spatial $\tau$-QPE) requires the definition of a related quantity: the spatial $\tau$-quantile causal effect (spatial $\tau$-QCE), which notably does not involve partial derivatives. The following definition and theorem provide the theoretical framework to recover the $\tau$-quantile partial effect across a spatial domain.

\begin{definition}[Spatial $\tau$-Quantile Causal Effect]\label{def:qce}
    Given the assumptions in Definition~\ref{def:qpe}, the \textit{spatial $\tau$-quantile causal effect} of $\mb X$ on $\mb Y$ can be defined as the function $g^{\tau}_{\textrm{QCE}(\mb X \to \mb Y)}: \mathbb{R}^{d+k} \to \mathbb{R}$, for every $\mb x \in \mathbb{R}^{d}$ and $\mb w \in \mathbb{R}^k$, given by
        \begin{align*}
             g^{\tau}_{\textrm{QCE}(\mb X \to \mb Y)}(\mbx, \mbw) & {:=} 
             \frac{1}{\vert S \vert}\int_{S}Q_{\tau}(\mb Y \mid \mb X = \mb x, \mb H = \mb h_{\bs s}^1, \mb W=\mb w)\mathrm{d}\bs s, \\
            &= \frac{1}{\vert S\vert}\int_{S} Q_{\tau}(f(\mb x, \mb h_{\bs s}^1, \mb w, \bs\eps_0^1))\mathrm{d}\bs s.
        \end{align*}
\end{definition}

\begin{theorem}[Causal Interpretation] \label{thm:1}
       Let $(\bs s, t) \in S \times \mathbb{N}$, with $S \subset \mathbb{R}^2$ as in \mbox{Definition~\ref{def:qpe}}, $\mb x \in \mathbb{R}^{d}$, and $\mbw \in \mathbb{R}^k$ be fixed, and consider any intervention on $\mb X$ such that $\mbX_{\bs s}^t = \mb x$ holds almost surely in the induced interventional distribution $\mathbb{P}_{\mb x}$. Denote the $\tau$-quantile of $Y_{\bs s}^t$ under $\PP_{\mb x}$ as $Q_{\tau}^{\PP_{\mb x}}(Y_{\bs s}^t)$ for a given $\bs s$ and $t$ and the spatially-integrated $\tau$-quantile of $\mb Y$ under $\PP_{\mb x}$ as 
       \[Q_{\tau}^{\PP_{\mb x}}[\mb Y] := \frac{1}{\vert S\vert}\int_{S}Q_{\tau}^{\PP_{\mb x}}(Y_{\bs s}^t)\mathrm{d}\bs s.\] We then have that \(Q_{\tau}^{\PP_{\mb x}}[\mb Y] = g^{\tau}_{\textrm{QCE}(\mb X \to \mb Y)}(\mb x,\mbw),\) that is, $g^{\tau}_{\textrm{QCE}(\mb X \to \mb Y)}(\mb x, \mbw)$ is the $\tau$-quantile of $Y_{\bs s}^t$ under any intervention that enforces $\mbX_{\bs s}^t = \mb x$.
    \end{theorem} The proof of Theorem~\ref{thm:1} is in Appendix~\ref{proof:1}.

\begin{corollary}
    Under the conditions of Theorem~\ref{thm:1}, we have that
    \[\frac{\partial Q_{\tau}^{\PP_{\mb x}}[\mb Y]}{\partial \mb x} = \frac{\partial g^{\tau}_{\textrm{QCE}(\mb X \to \mb Y)}(\mb x,
    \mbw)}{\partial \mb x} = g^{\tau}_{\textrm{QPE}(\mb X \to \mb Y)}(\mb x, \mbw);\] 
    that is, $g^{\tau}_{\textrm{QPE}(\mb X \to \mb Y)}(\mb x,\mbw)$ is the partial derivative with respect to $\mb x$ of the $\tau$-quantile of $Y_{\bs s}^t$ under any intervention that enforces $\mbX_{\bs s}^t = \mb x$.
\end{corollary} 
Due to the exchangeability of the integral and partial derivatives, the proof is trivial. 
We now show that $\qestimand$ is identifiable through the full observational distribution.
\begin{theorem}[Identifiability]
     Let $g^{\tau}_{Y \mid (\mb X, \mb H, \mb W)}$ denote the quantile regression function $(\mb x, \mb h, \mb w) \mapsto Q_{\tau}(Y_{\bs s}^t \mid \mbX_{\bs s}^t = \mb x, \mbH_{\bs s}^t = \mb h, \mbW_{\bs s}^t = \mb w)$. 
     We have that, for all $\mb x \in \bb{R}^{d}$ and $\mb w \in \bb{R}^k$,
        \begin{equation}
        g^{\tau}_{\textrm{QPE}(\mb X \to \mb Y)}(\mb x,\mbw) = \frac{1}{\vert S\vert }\int_{S} \frac{\partial g^{\tau}_{Y \mid (\mb X, \mb H, \mb W)}(\mb x, \mb h_{\bs s}^1, \mb w)}{\partial \mb x} \mathrm{d}\bs s.
        \label{eq:quantregest}
        \end{equation}
\label{thm:2} \vspace{-1cm}
\end{theorem} The proof of Theorem~\ref{thm:2} is in Appendix~\ref{proof:2}.

\subsection{Estimation of the Spatial $\tau$-QPE}
\label{subsec:estimation}
We now define the following estimator for the spatial $\tau$-QPE. Consider the dataset $$(\bs{X}_n^{m_n}, \bs{Y}_n^{m_n}, \bs{W}_n^{m_n}) =  \left\{(\mbX_{\bs s}^t, Y_{\bs s}^t, \mbW_{\bs s}^t), \bs s \in \{\bs s_1, \dots, \bs s_n\}, \text{and for each } \bs s_i, t \in \{1, \dots, m_i \}\right\}.$$ Our proposed estimator, defined in Definition~\ref{def:estimator} below, is based on first fitting a quantile regression model to all $m_i$ samples at each location $\bs s_i$, and then integrating over the spatial domain. More precisely, given spatial locations $\bs s_1, \dots, \bs s_n$, we approximate the integral in \eqref{eq:def} using a discretization of the domain $S$ with each location $\bs s_i \in S_i, \bigcup_i S_i = S$.
    
\begin{definition}[Spatial $\tau$-QPE Estimator] \label{def:estimator}
    We denote
    \begin{equation}
        \hat{g}^{nm,\tau}_{\textrm{QPE}(\mb X \to \mb Y)}(\bs{X}_n^{m_n}, \bs{Y}_n^{m_n}, \bs{W}_n^{m_n})(\mb x,\mbw) := \frac{1}{\sum_i \vert S_i\vert}\sum_{i=1}^n \frac{\partial \hat{g}^{m_i, \tau}_{Y \mid (\mb X, \mb W)}(\bs{X}_{\bs s_i}^{m_i}, \bs{Y}_{\bs s_i}^{m_i}, \bs{W}_{\bs s_i}^{m_i})(\mb x, \mb w)}{\partial \mb x} \times \vert S_i\vert,
        \label{eq:estimator}
    \end{equation} where $\hat{g}^{m_i, \tau}_{Y \mid (\mb X, \mb W)}(\bs{X}_{\bs s_i}^{m_i}, \bs{Y}_{\bs s_i}^{m_i}, \bs{W}_{\bs s_i}^{m_i})(\mb x, \mb w)$ ($\tau$-QPE at location $\bs s_i$) is any suitable estimator of the $\tau$-quantile regression model with $m_i$ samples, and $\vert S_i\vert$ is the area of $S_i$.
\end{definition}
\begin{remark}
    The proposed estimator is based on the idea that, for every $\bs s_i \in \{\bs s_1, \dots, \bs s_n\} \subseteq S$, we observe several time instances $(\mbX_{\bs s_i}^t, Y_{\bs s_i}^t, \mbW_{\bs s_i}^t), t \in \{1, \ldots, m_i\}$, with the same conditionals, $Y^t_{\bs s_i} \mid (\mbX_{\bs s_i}^t, \mbH_{\bs s_i}^t, \mbW_{\bs s_i}^t)$. Since $\mb H$ is time-invariant, we can, for every $\bs s_i$, estimate $g^{\tau}_{Y \mid (\mb X, \mb H, \mbW)}(\cdot, \mbh_{\bs s_i}^1, \cdot)$ for the (unobserved) realization $\mb h_{\bs s_i}^1$ of $\mbH_{\bs s_i}^1$ using the data $(\mbX_{\bs s_i}^t, Y_{\bs s_i}^t, \mbW_{\bs s_i}^t)$, $t \in \{1, ..., m_i\}$. The integral~\eqref{eq:quantregest} is then approximated by averaging estimates obtained from different spatial locations, which is given by~\eqref{eq:estimator}.
\end{remark}
Based on \eqref{eq:estimator}, we now establish that our estimator is consistent under mild assumptions.

\begin{assumption}[Consistent estimation of the $\tau$-QPE at location $\bs s_i$]
   \label{ass:1} 
   There exists a non-empty set $\mc{X} \times \mc{W} \subset \bb{R}^{d} \times \bb{R}^{k}$, such that for all $\mb x \in \mc{X}, \mbw \in \mc{W}$, and $\bs s_i \in S \subseteq \bb{R}^2$, $i = 1, \dots, n,$
        \[\partial \hat{g}^{m_i, \tau}_{Y \mid (\mb X, \mb W)}(\bs{X}_{\bs s_i}^{m_i}, \bs{Y}_{\bs s_i}^{m_i}, \bs{W}_{\bs s_i}^{m_i})(\mb x, \mb w)/\partial \mbx - \partial g^{\tau}_{Y \mid (\mb X, \mb H, \mb W)}(\mb x, \mb h_{\bs s_i}^1, \mb w)/\partial \mbx \to 0,\] in probability as $m_i \to \infty$.
\end{assumption}
\begin{remark}
    Assumption~\ref{ass:1} ensures that, as the sample size $m_i$ for site $\bs s_i$ increases, the estimator of the $\tau$-QPE at each specific site $\bs s_i$ converges to the true underlying conditional quantile function with the hidden confounder $\mbH$ included.
\end{remark}


\begin{theorem}[Consistent estimation of the spatial $\tau$-QPE]
\label{thm:3}
        Let \((\mb Y, \mb X, \mb H, \mb W)\) follow the QLSCM defined in Definition~\ref{def:qlscm}. Let \((\bs s_n)_{n \in \mathbb{N}}\) be a sequence of spatial coordinates such that the \mbox{marginalized} process \((\mbH_{\bs s_n}^1)_{n \in \mathbb{N}}\) has realizations $(\mb h_{\bs s_n}^1)_{n \in \mathbb{N}}$, and the spatial coordinates $\bs s_i$ coincide with the \mbox{centroids} of the gridded blocks $S_i$, where $\bigcup_i S_i = S$. Let \(\hat{g}^{\tau}_{Y \mid (\mb X, \mb W)} = (\hat{g}^{m_i, \tau}_{Y \mid (\mb X, \mb W)})_{m_i \in \mathbb{N}}\) be an estimator satisfying Assumption~\ref{ass:1}. We then have the following consistency result: for all \(\mbx \in \mathcal{X}\), \(\delta > 0\), and \(\alpha > 0\), there exists \(N \in \mathbb{N}\) such that, for all \(n \geq N\), there exists \(M_n \in \mathbb{N}\) such that for all \(m_1, \ldots, m_n \geq M_n\) we have that
\[
\PP\left(\left\|\hat{g}^{nm, \tau}_{\textrm{QPE}(\mb X \to \mb Y)}(\bs{X}_n^{m_n}, \bs{Y}_n^{m_n}, \bs{W}_n^{m_n})(\mb x,\mbw) - g^{\tau}_{\textrm{QPE}(\mb X \to \mb Y)}(\mb x,\mbw)\right\| > \delta \right) \leq \alpha.
\]
\end{theorem} The proof of Theorem~\ref{thm:3} is given in Appendix~\ref{proof:3}.
\begin{remark}
Assumption~\ref{ass:1} makes a convenient consistency assumption for the site-specific $\tau$-QPE. This assumption largely depends on the choice of a suitable quantile regression estimator and mild conditions on the underlying model structure of $Y$ given $\mbX, \mbH, \mbW$; see Theorem~\ref{thm:4}. 
\end{remark}

\section{Simulation Study}
\label{sec:5}

In this section, we conduct simulation studies comparing our proposed spatial $\tau$-QPE estimator under the QLSCM framework with the average causal effect (ACE) estimator from \citet{christiansen_toward_2022} under the classical LSCM framework. 
To evaluate our method's capability in estimating the marginal causal effect (causal partial effect) of the exposure on the outcome, we present three illustrative cases across two scenarios (explicit and implicit confounding):\\
\textbf{\underline{Case~1}}: Estimation of a \textbf{constant} marginal causal effect under explicit hidden confounding.\\
\textbf{\underline{Case~2}}: Detection of a \textbf{null} marginal causal effect under implicit hidden confounding.\\
\textbf{\underline{Case~3}}: Estimation of \textbf{heterogeneous} marginal causal effects under implicit hidden confounding.

An explicit confounder appears directly as a variable in the outcome's structural equation, whereas an implicit confounder influences the outcome indirectly by governing the parameters of intermediate latent variables, that affect the distribution of the outcome (see Appendix~\ref{app:confounders}).

\subsection{General Simulation Setting}

We sample spatio-temporal data over $(\bs s_i)_{i \in \{1,\ldots,1000\}}$, a regular $50 \times 20$ spatial grid, with $m_i = 100$ samples, $i=1,\dots,1000$, at each location.

We conduct only one experiment for Case~1 because its purpose is to validate the bias-removal capability of our QLSCM framework against the LSCM, not to measure estimation accuracy. Conversely, to robustly evaluate the accuracy of our causal effect estimates in Cases 2 and 3, we repeat each experiment $100$ times. 
Our simulation study compares the following estimators:
\begin{enumerate}
    \item the average causal effect (ACE) estimator, defined as: 
    \begin{equation}
    (\beta_0^*, \beta_1^*) := \hat{f}_{\text{ACE}(\mbX \to \mbY)}^{nm}(\mbX_n^m, \mbY_n^m)(x) = \frac{1}{n} \sum_{i=1}^n (1 \; x) \hat{\beta}_{\text{OLS}}^m(\mbX_{\bs s_i}^m, \mbY_{\bs s_i}^m); \label{eq:acenew}
    \end{equation}
    \item the spatial $\tau$-QPE estimator with $\tau = 0.5$, defined as: 
    \begin{equation}
        (\beta_0^{\tau}, \beta_1^{\tau}) := \hat{g}_{\text{QPE}(\mbX \to \mbY)}^{nm,\tau}(\mbX_n^{m}, \mbY_n^{m})(x) = \frac{1}{n} \sum_{i=1}^n (1 \; x) \hat{\beta}_{\text{QR}}^{m,\tau}(\mbX_{\bs s_i}^{m}, \mbY_{\bs s_i}^{m}); \label{eq:qpenew}
    \end{equation} 
    \item a generalised additive model (GAM) estimator~\citep{hastie2017generalized}  (used only for Case~1),
\end{enumerate} where $\hat{\beta}_{\text{OLS}}^m \in \mathbb{R}^2$ and $\hat{\beta}_{\text{QR}}^{m,\tau} \in \mathbb{R}^2$ are the linear regression and quantile regression coefficient estimates, respectively. These are both obtained by fitting the model at each spatial location $\bs s_i$ based on outcome data $\mbY_{\bs s_i}^{m}=\left\{Y_{\bs s_i}^t\right\}_{t=1}^m$ and exposure $\mbX_{\bs s_i}^{m}=\left\{X_{\bs s_i}^t\right\}_{t=1}^m$.

\subsection{Simulation Results}

\paragraph{Case~1.}
\label{case1}
Following Example 1 of \citet{christiansen_toward_2022}, we consider independent realizations $\boldsymbol{\gamma}, \boldsymbol{\phi}, \boldsymbol{\xi}^t$, $\boldsymbol{\delta}^t$, for $t \in \mathbb{N}$, drawn from a univariate stationary spatial Gaussian process centered at zero with covariance function $\exp(-\frac{1}{2}\|\bs s_i - \bs s_j\|_2)$ for arbitrary locations $\bs s_i$ and $\bs s_j$, where $\|\cdot\|_2$ denotes the Euclidean norm. The joint distribution is characterized through a marginal distribution for the hidden confounder $\mathbf{H}$---modeled as a bivariate process with components $\widetilde{\mbH}$ and $\widehat{\mbH}$---and conditional distributions $\mathbf{X} \mid \mathbf{H}$ and $\mathbf{Y} \mid (\mathbf{X}, \mathbf{H})$, specified for each $(\bs s,t) \in \mathbb{R}^2 \times \mathbb{N}$ as:
\begin{align}
\mathbf{H}_{\bs s}^t &:= (\tilde{H}_{\bs s}^t, \hat{H}_{\bs s}^t)^\top = \left({\gamma}_{\bs s}, \frac{1}{2}{\xi}_{\bs s} + \frac{\sqrt{3}}{2}{\phi}_{\bs s}\right)^\top, \nonumber\\
X_{\bs s}^t &= \exp(-\|\bs s\|_2^2/1000) + (0.2 + 0.1 \cdot \sin(2\pi t/100)) \cdot \tilde{H}_{\bs s}^t \cdot \hat{H}_{\bs s}^t + 0.5 \cdot {\xi}_{\bs s}^{t}, \nonumber\\
Y_{\bs s}^t &= (1.5 + \tilde{H}_{\bs s}^t \cdot \hat{H}_{\bs s}^t) \cdot X_{\bs s}^t + (\tilde{H}_{\bs s}^t)^2 + |\hat{H}_{\bs s}^t| \cdot \delta_{\bs s}^{t}. \label{eq:ex1_structure}
\end{align} For clarity of exposition, we denote $\tilde{\mathbf{H}}$ and $\hat{\mathbf{H}}$ as the first and second components of the bivariate process $\mathbf{H}$, respectively. Interventional distributions for $\mbX$, $\mbY$, or $\mathbf{H}$ follow the causal graphical model which is discussed in Section~\ref{sec:causal_graph} (as do Cases 2 and 3, to be introduced). In this case, our interest is to estimate a \textbf{constant} marginal causal effect of $\mbX$ on $\mbY$ ($\beta_1$) that is endowed by a linear model $x \mapsto \beta_0 + \beta_1 x$, with $\beta_0 =  \mathbb{E}[(\tilde{H}_{\bs 0}^1)^2 + |\hat{H}_{\bs 0}^1|\cdot \delta_{\bs 0}^1] = 1$ and $\beta_1 = 1.5 + \mathbb{E}[\tilde{H}_{\bs 0}^1 \cdot \hat{H}_{\bs 0}^1] = 1.5$. 

Observe that the coefficient of $X_{\bs{s}}^t$ in \eqref{eq:ex1_structure} contains the product $\tilde{H}_{\bs{s}}^t \cdot \hat{H}_{\bs{s}}^t$, leading to heterogeneous spatial QPE estimates across different values of $\tau$. As noted, this example serves to demonstrate the elimination of hidden confounding bias. The heterogeneity term $\tilde{H}_{\bs{s}}^t \cdot \hat{H}_{\bs{s}}^t$ exhibits symmetry around zero, being the product of zero-mean Gaussian variables. Since the median and mean coincide for symmetric distributions, we thus select $\tau = 0.5$ for spatial QPE estimation.
\begin{figure}[t!]
    \centering
    \includegraphics[width=0.45\linewidth]{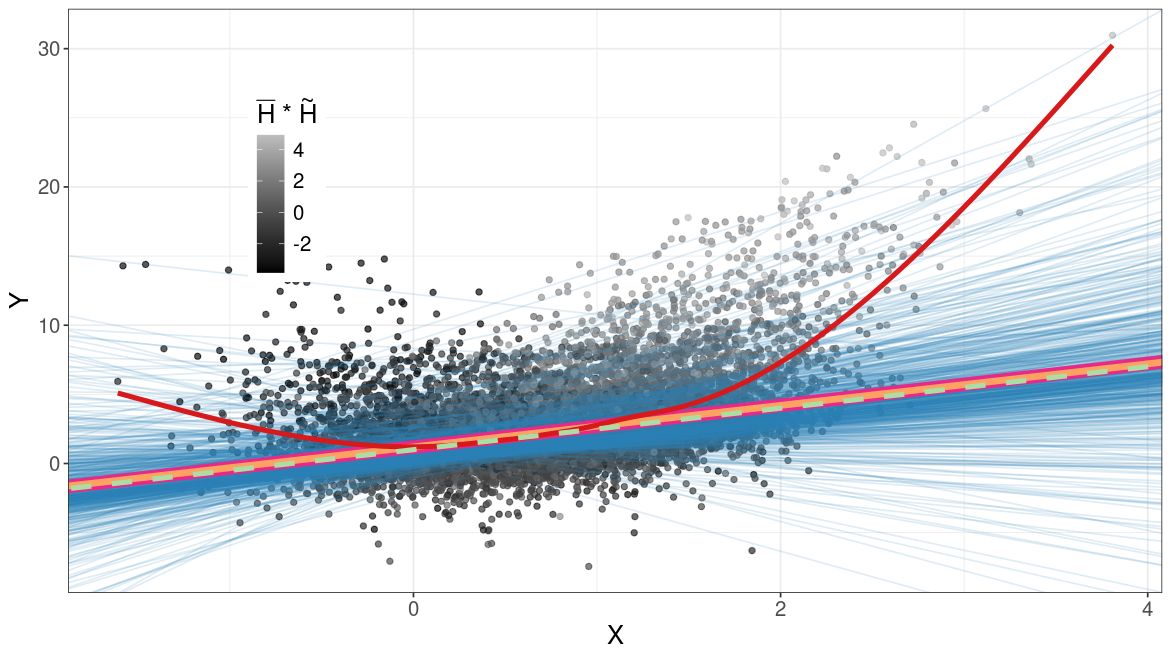}
    \caption{Greyscale points display the scatter plot of the confounding term $\bar{\mathbf{H}} \cdot \tilde{\mathbf{H}}$ from Case~1. The true marginal causal effect $\beta_1$ is represented by the slope of the dashed green line. ACE estimates appear as the thick solid magenta line. Site-wise quantile estimates $g_{\text{QPE}(\mathbf{X} \rightarrow \mathbf{Y})}^{m, 0.5}(\mbx)$ (as defined in Assumption~\ref{ass:1}) are shown as thin blue lines, while the spatially-aggregated QLSCM estimate at $\tau = 0.5$, from Equation~\eqref{eq:def}, appears as the solid orange line. A semi-parametric GAM smoother fitted to the regression function is depicted by the solid red curve.} 
    \label{fig:bias-correction}
\end{figure}

\begin{hyphenrules}{nohyphenation}
The results are presented in Figure~\ref{fig:bias-correction}, where the background scatter plot reveals the \mbox{confounding} structure induced by $\tilde{\mathbf{H}} \cdot \hat{\mathbf{H}}$, with grayscale values indicating the magnitude of this biasing term. The site-wise quantile estimates exhibit substantial heterogeneity, \mbox{reflecting} the spatially-varying confounding effects. Notably, the QLSCM spatially-aggregated median \mbox{estimate} and LSCM-based ACE estimate closely \mbox{approximate} the true causal effect $\beta_1$, demonstrating the methods' \mbox{effectiveness} in removing confounding bias via spatial aggregation. The GAM smoother, which does not account for hidden confounders, falsely captures the nonlinear causal relationship in the observed data. The contrast between the GAM curve and the (Q)LSCM estimates illustrates how traditional smoothing methods fail to recover the true causal relationship, while the proposed spatial quantile approach successfully identifies the underlying linear effect, despite the presence of hidden confounders, just as the ACE estimator.
\end{hyphenrules}

\paragraph{Case~2.}
\label{case2}
Consider again process $ \bs \gamma, \bs\xi^t, t \in \bb{N}$. The marginal process $\xi_{\bs s}^t$ follows a stationary Gaussian distribution with mean zero, and we transform the margins of $\gamma$ to $\text{Unif}(10, 20)$. We then define the data generation model $\mbY \mid (\mbX, \mbH)$, specifying that for all $(\bs s,t) \in \bb{R}^2 \times \bb{N}$:
\begin{align*}
    H_{\bs s}^t &= \gamma_{\bs s},\qquad {\eps}_{\bs s}^t \overset{\text{IID}}{\sim} \text{GEV}(0, 1, 1.2), \qquad X_{\bs s}^t = 1 + 5\sin(2\pi t/100) + \xi_{\bs s}^t, \\
    Y_{\bs s}^t &\sim 50 + 2 \mc{N}(-0.5X_{\bs s}^t + H_{\bs s}^t, 0.5) + X_{\bs s}^t + 2H_{\bs s}^t + 2\eps_{\bs s}^t.
\end{align*} Here, GEV($\mu$, $\sigma$, $\xi$) denotes the generalized extreme-value distribution \citep{davison2015statistics} whose distribution function is given by $
F(x; \mu, \sigma, \xi) = 
\exp\left( -\left[1 + \xi (x - \mu)/{\sigma}\right]^{-{1/\xi}}\right)$, for $1 + \xi ({x - \mu})/{\sigma} > 0, \xi > 0.
\) When $\xi>0$, the heavy-tailedness of the GEV distribution may cause traditional regression methods to fail; note that the conditional expectation is infinite when $\xi \geq 1$. 

\begin{hyphenrules}{nohyphenation}
We compare two estimates of the marginal causal effect of $\mbX$ on $\mbY$: 1) $\beta_1^*$ from the classical LSCM (Equation~\eqref{eq:acenew}) and 2) $\beta^{\tau}_1$ from the spatial $0.5$-QPE using our proposed QLSCM framework (Equation~\eqref{eq:qpenew}). The underlying ground truth of the marginal causal effect of $Y_{\bs s}^t$ on $X_{\bs s}^t$ is the constant zero, as $2 \times (-0.5 X_{\bs s}^t) + X_{\bs s}^t = 0$. Thus, the causal influence of $\mbX$ on $\mbY$ is zero since there is no term relating to $\mbX$. 
\end{hyphenrules}
Given that the underlying error term $\eps_{\bs s}^t$ exhibits heavy-tailed behavior with tail index $\xi = 1.2$, the conditional expectation of $\mbY$ given $\mbX$ is infinite. Consequently, the linear regression estimator under the classical LSCM, represented by $\beta_1^*$, produces incorrect estimates, as shown by the blue dashed line in the left panel of Figure~\ref{simulation_1}. In contrast, our proposed quantile-based estimator $\beta_1^{\tau}$ successfully recovers the true marginal causal effect, evidenced by the slope of the red dashed line which perfectly coincides with the ground truth (zero; green dashed line). The slight difference in intercepts arises because the conditional quantile estimator does not converge to a constant value as $\tau$ varies. The right panel of Figure~\ref{simulation_1} presents results from $100$ simulations and displays boxplots of the logarithm of the absolute error between the estimated and the true marginal causal effects. We find that the spatial $\tau$-QPE achieves superior performance with much smaller absolute errors, compared to the classical ACE.

\begin{figure}[t!]
    \centering
    \begin{minipage}{0.49\textwidth}
        \centering
        \includegraphics[width=\textwidth]{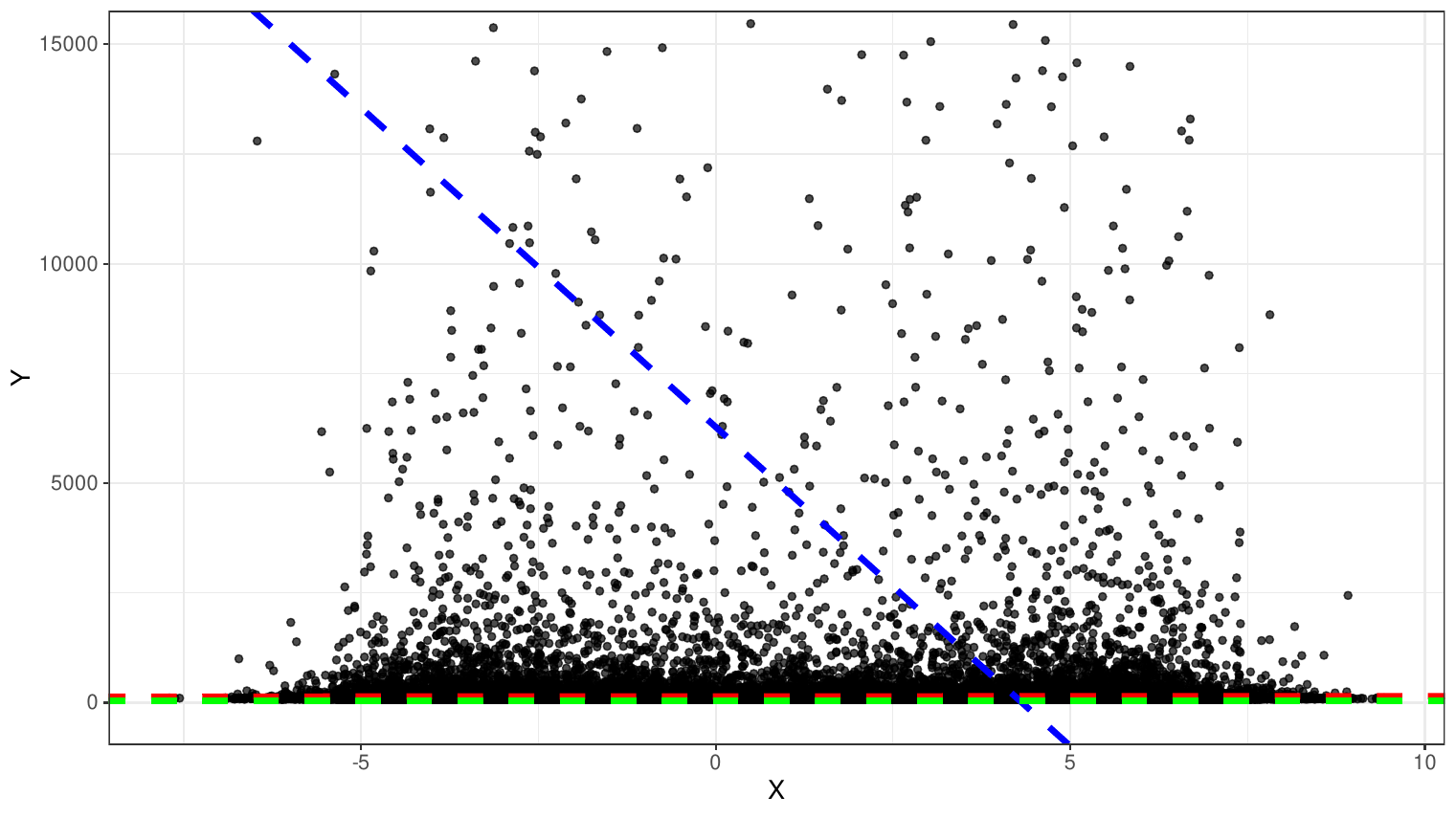}
    \end{minipage}
    \hfill
    \begin{minipage}{0.49\textwidth}
        \centering
        \includegraphics[width=\textwidth]{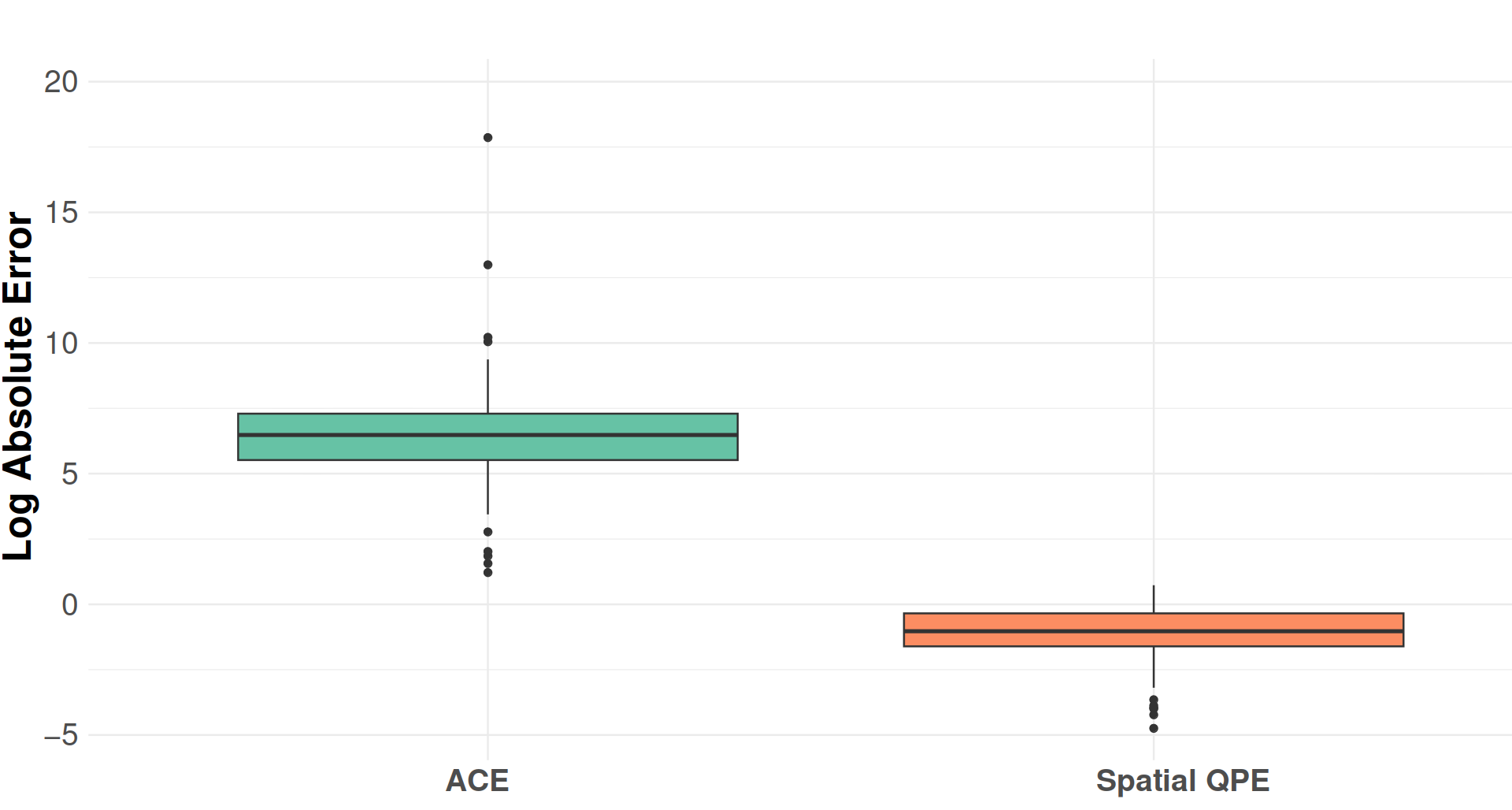}
    \end{minipage}
    \caption{Illustrative plots for Case~2. \textbf{Left:} Data (black dots) and true/estimated inferential target (marginal causal effect). The slope of the dashed lines represent the spatial $\tau$-QPE (red), the ground truth (green), and the marginal ACE derived from the LSCM (blue). The ground truth for ACE and spatial $\tau$-QPE are identical (coincident) in this setting, where that theoretical value is zero. \textbf{Right:} Boxplot of log absolute errors between the true marginal causal effect and the estimated effect using the ACE (LSCM; green) and the spatial $0.5$-QPE (QLSCM; orange).}
    \label{simulation_1}
\end{figure}

\paragraph{Case~3.}
\label{case3}
Recall $\bs \gamma$ from Case~2, except now with marginal $\gamma_{\bs s} \sim \text{Unif}(50, 100)$ for all sites $\bs s$. We define the data generation model $\mbY \mid (\mbX, \mbH)$ specifying that for all $(\bs s,t) \in \bb{R}^2 \times \bb{N}$ marginally:
\begin{align*}
    H_{\bs s}^t &= \gamma_{\bs s}, \qquad X_{\bs s}^t = 5 [\exp(-\|\bs s\|_2/1000) + (1 + 5\sin(2\pi t/100)) + \xi_{\bs s}^t], \\
    Y_{\bs s}^t &\sim 4 + 2 \mathcal{N}(-0.5X_{\bs s}^t + H_{\bs s}^t, 0.5) + 2 \Gamma(0.005(X_{\bs s}^t + H_{\bs s}^t), 0.01),
\end{align*} where $\Gamma(\alpha, \lambda)$ denotes the Gamma distribution with shape $\alpha>0$ and rate $\lambda>0$. We compare two estimates: 1) $\beta_1^*$ from Equation~\eqref{eq:acenew}; 2) $\beta^{\tau}_1$ from Equation~\eqref{eq:qpenew} with $\tau \in (0, 1)$. In this example, the Gamma distribution in the data generation of $Y_{\bs s}^t$ does not belong to a location-scale family and its changing shape will lead to heterogeneous marginal causal effects across different $\tau$ values. The true marginal causal effect targeted by the spatial $\tau$-QPE is defined as the slope $\beta_1^{\tau}$, obtained by quantile regression conditioned on the full latent confounders $\mbH$ and covariate $\mbX$. Since $\mbH$ is unobserved in practice, this serves as a theoretical benchmark. Notably, the ACE estimated via the LSCM yields a constant marginal effect, which is invariant with respect to $\tau$.

Figure~\ref{fig:case3} compares the estimated marginal causal effect $\hat{\beta}$ of $\mathbf{X}$ on $\mathbf{Y}$ against the ground truth across quantile levels $\tau$ for Case 3. 
\begin{figure}[t!]
    \centering
    \includegraphics[width=0.5\linewidth]{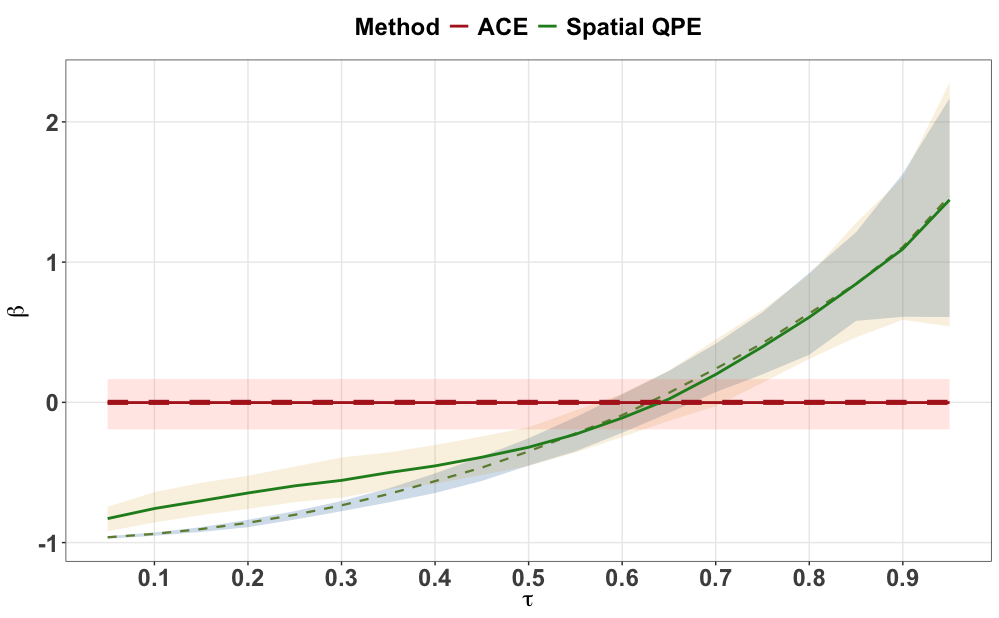}
    \caption{Estimated marginal causal effect $\hat\beta$ of $\mathbf{X}$ on $\mathbf{Y}$ across quantile levels $\tau$ for Case 3. The spatial $\tau$-QPEs ($\beta_1^\tau$; QLSCM) are shown with ground truth (green dashed line, 95\% confidence band in light blue) and estimates (green solid line, 95\% confidence band in gold). The ACEs ($\beta_1^*$; LSCM) are shown with ground truth (red dashed line) and estimates (red solid line, 95\% confidence band in light red). All confidence bands are constructed from 100 simulations.}
    \label{fig:case3}
\end{figure}
While the ACE is zero, the spatial $\tau$-QPE  successfully reveals the heterogeneous marginal causal effects.
A modest bias is observed at lower quantiles ($\tau < 0.35$), where the estimated effect slightly exceeds the ground truth. This bias stems from the heteroskedastic confounding induced by the Gamma-distributed component in the data-generating process. The Gamma distribution's bounded support at zero creates a floor effect, causing the conditional quantile function to exhibit higher curvature near the lower tail. This nonlinearity leads to biased estimates when fitting a linear quantile regression. Nevertheless, this bias is inconsequential for inference: the QLSCM correctly identifies both the sign and statistical significance of the marginal causal effect throughout the quantile range, as evidenced by the confidence bands. At higher quantiles, the unbounded tail exhibits approximately linear behavior, yielding homogeneous local effects that spatial aggregation successfully recovers. Overall, the QLSCM demonstrates robust performance in detecting heterogeneous causal effects that the constant ACE approach misses.

\section{Causal Analysis of US Wildfires on Air Quality}
\label{sec:6}
\subsection{Overview}

\begin{hyphenrules}{nohyphenation}
We investigate the causal relationship between fire radiative power (FRP) and aerosol optical depth (AOD) while controlling for observed confounders i) enhanced vegetation index (EVI), ii) precipitation (PRCP), iii) near-surface wind speed (WS), and iv) maximum temperature (TMAX) (recall Section~\ref{sec:data}), by employing the spatial $\tau$-QPE within the QLSCM framework. We compare estimates of the spatial $\tau$-QPE with the ACE, focusing on spatial aggregation over distinct climate regions (Section~\ref{sec:regional}) and individual states (Section~\ref{sec:state-wise}). We consider quantile levels $\tau = 0.1, 0.5, 0.9$ to examine the left tail, median, and right tail of the distribution, where $\tau = 0.5$ enables comparison with the LSCM-based ACE. Our approach reveals the structure of the underlying conditional distribution, indicating that heterogeneous treatment effects across quantiles exist. To provide additional details on heterogeneous treatment effects and conduct a more focused analysis, we select representative states for further examination (see Section~\ref{sec:rep_states}).
\end{hyphenrules}

\begin{hyphenrules}{nohyphenation}
Following, e.g., \citet{richards2022regression}, we focus on June to October (summer and fall) as they represent the peak wildfire season. An exploratory analysis, applying Hill's estimator \citep{resnick2007heavy} to the pooled AOD data, yields strictly positive estimates of the tail index, confirming that the AOD distribution is indeed heavy-tailed. A critical implication of this tail behavior is that standard statistical procedures relying on finite moment assumptions are likely to fail. This empirical observation highlights the necessity of our proposed approach and directly motivates the design of Case~2 in our simulation study, which evaluates model performance under these extreme conditions. To quantify uncertainty, we develop a hypothesis test to determine whether the marginal causal effect of $\mathbf{X}$ on $\mathbf{Y}$ is zero by considering the null hypothesis $H_0$ that $(\mathbf{X}, \mathbf{Y}, \mathbf{W})$ come from a QLSCM with a function $f$ that is constant with respect to $\mathbf{X}_{\boldsymbol{s}}^{t}$ for all $(\boldsymbol{s},t) \in \mathbb{R}^2 \times \{1, 2, \ldots\}$. To this end, we employ a stationary bootstrap \citep{politis1994stationary} with expected block size of $5$ days. Each bootstrap sample is constructed by repeatedly drawing uniformly-random starting times and random block sizes from a geometric distribution with expectation $5$ until obtaining sufficient temporal coverage. Using $250$ bootstrap samples, we build $99\%$ confidence intervals to test $H_0$. Choosing a level of $99$\% yields conservative results and reduces the chance of mistakenly claiming an nonexistent effect.
\end{hyphenrules}

\subsection{Daily Regional Analysis}
\label{sec:regional}
\begin{hyphenrules}{nohyphenation}
We follow \citet{chen2021effects} and divide the CONUS into nine climate regions: Northwest, West, Southwest, Northern Rockies and Plains, Upper Midwest, Ohio Valley, Northeast, South, and Southeast. We apply the (Q)LSCM approach within each region. Figure~\ref{figure:44} reveals that the marginal causal effect of FRP on AOD varies substantially across quantiles and regions. The most pronounced and statistically significant effects occur at the upper tail ($\tau = 0.9$), particularly in the West region (California area), which exhibits strong positive causal effects. Lower quantiles ($\tau = 0.1, 0.5$) show weaker and less significant effects across all regions. The LSCM-based ACE (Figure~\ref{fig:8-4}) only estimates the conditional mean of $\mbY \mid \mbX$. The observed pattern suggests that wildfire intensity primarily impacts air quality during extreme pollution events, with the western and northwestern United States being most susceptible to these causal relationships. While the ACE estimates a single average effect, which shows similar patterns to the spatial $0.5$-QPE, the spatial $\tau$-QPE at level $\tau = 0.1$ and $\tau = 0.9$ reveal that the impact of FRP on AOD is not constant, becoming dramatically stronger at higher pollution levels. Further, ACE estimates do not reveal significant causal partial effects in the Northwest region whereas the spatial $0.9$-QPE does.
\end{hyphenrules}

\begin{figure}[!t]
    \centering
    \setlength{\abovecaptionskip}{2pt}
    \setlength{\belowcaptionskip}{0pt}
    
    \begin{subfigure}{0.4\textwidth}  
        \centering
        \includegraphics[width=\linewidth]{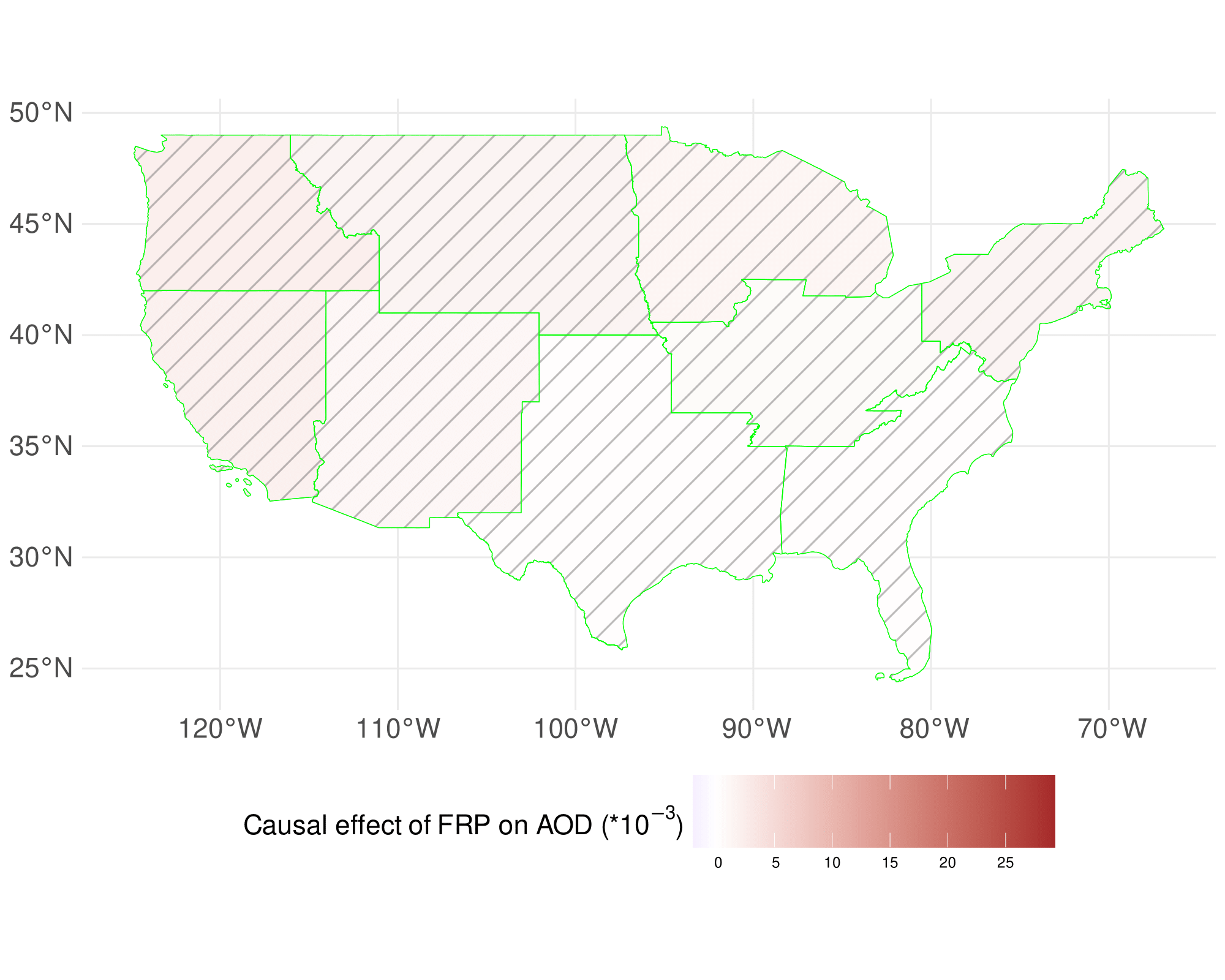}
        \vspace{-60pt}  
        \caption{$\tau = 0.1$}
        \label{fig:8-1}
    \end{subfigure}
    \hspace{1.5cm}
    \begin{subfigure}{0.4\textwidth}  
        \centering
        \includegraphics[width=\linewidth]{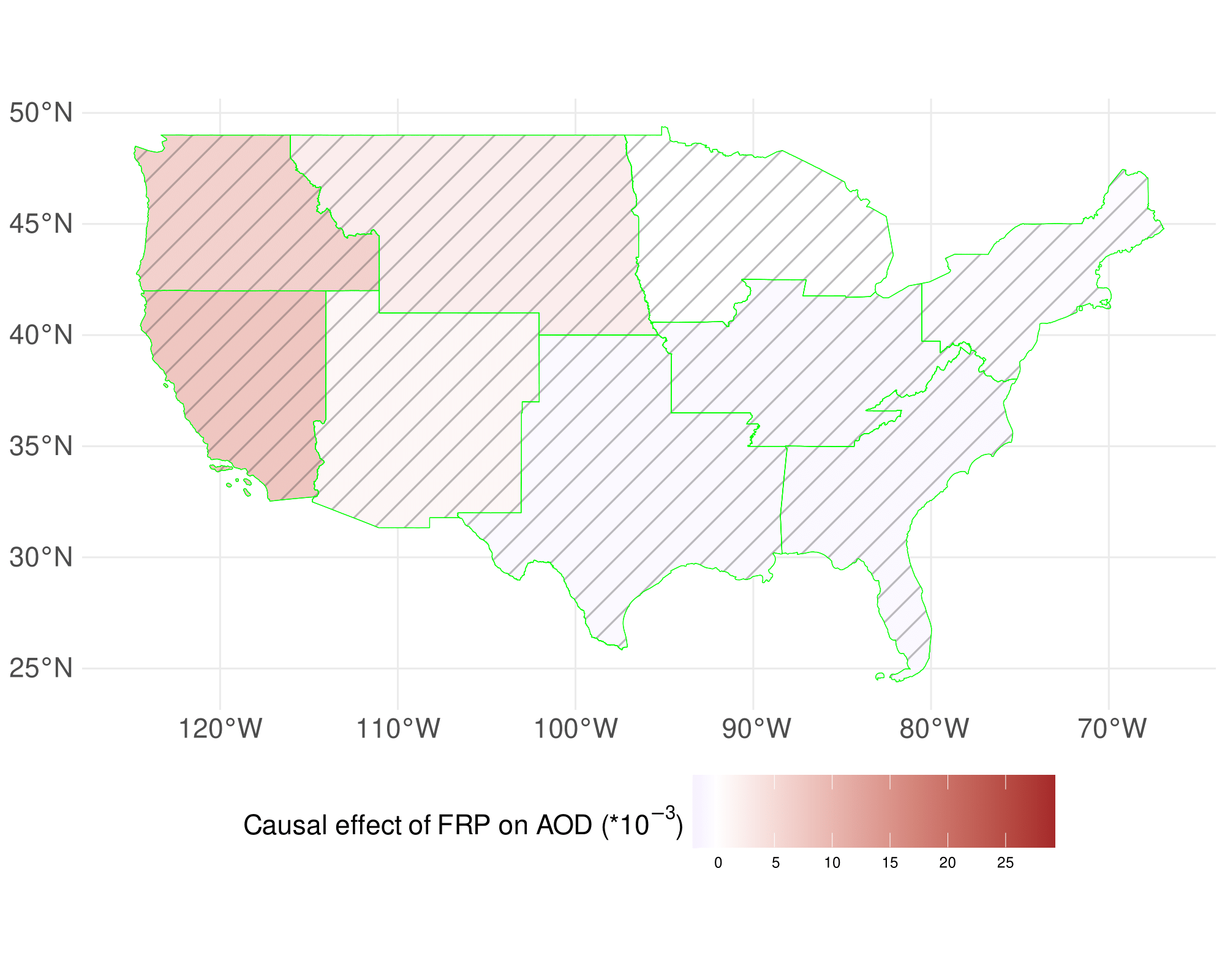}
        \vspace{-60pt}  
        \caption{$\tau = 0.5$}
        \label{fig:8-2}
    \end{subfigure}

    \vspace{-8pt}  
    
    \begin{subfigure}{0.4\textwidth}  
        \centering
        \includegraphics[width=\linewidth]{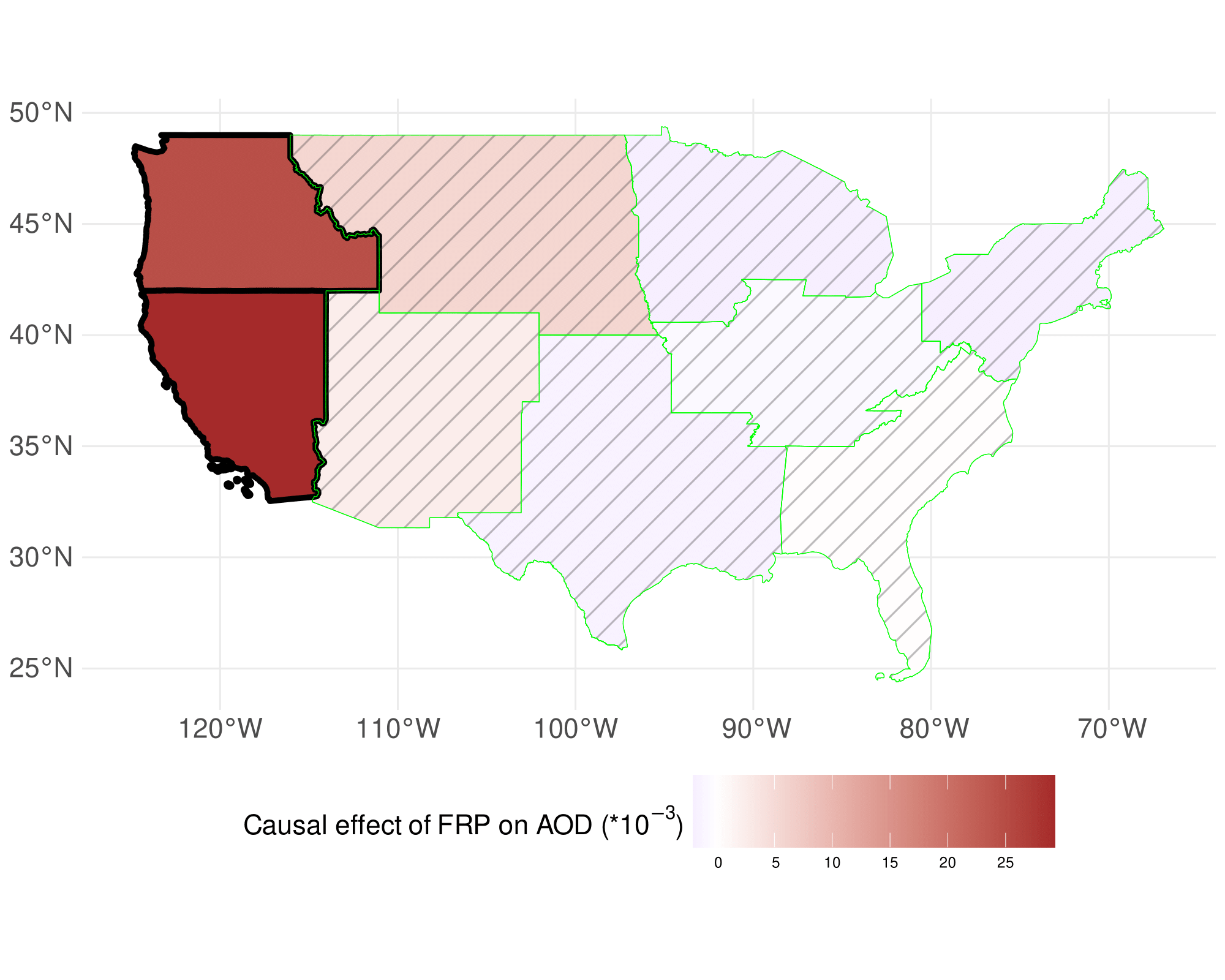}
        \vspace{-60pt}  
        \caption{$\tau = 0.9$}
        \label{fig:8-3}
    \end{subfigure}
    \hspace{1.5cm}
    \begin{subfigure}{0.4\textwidth}  
        \centering
        \includegraphics[width=\linewidth]{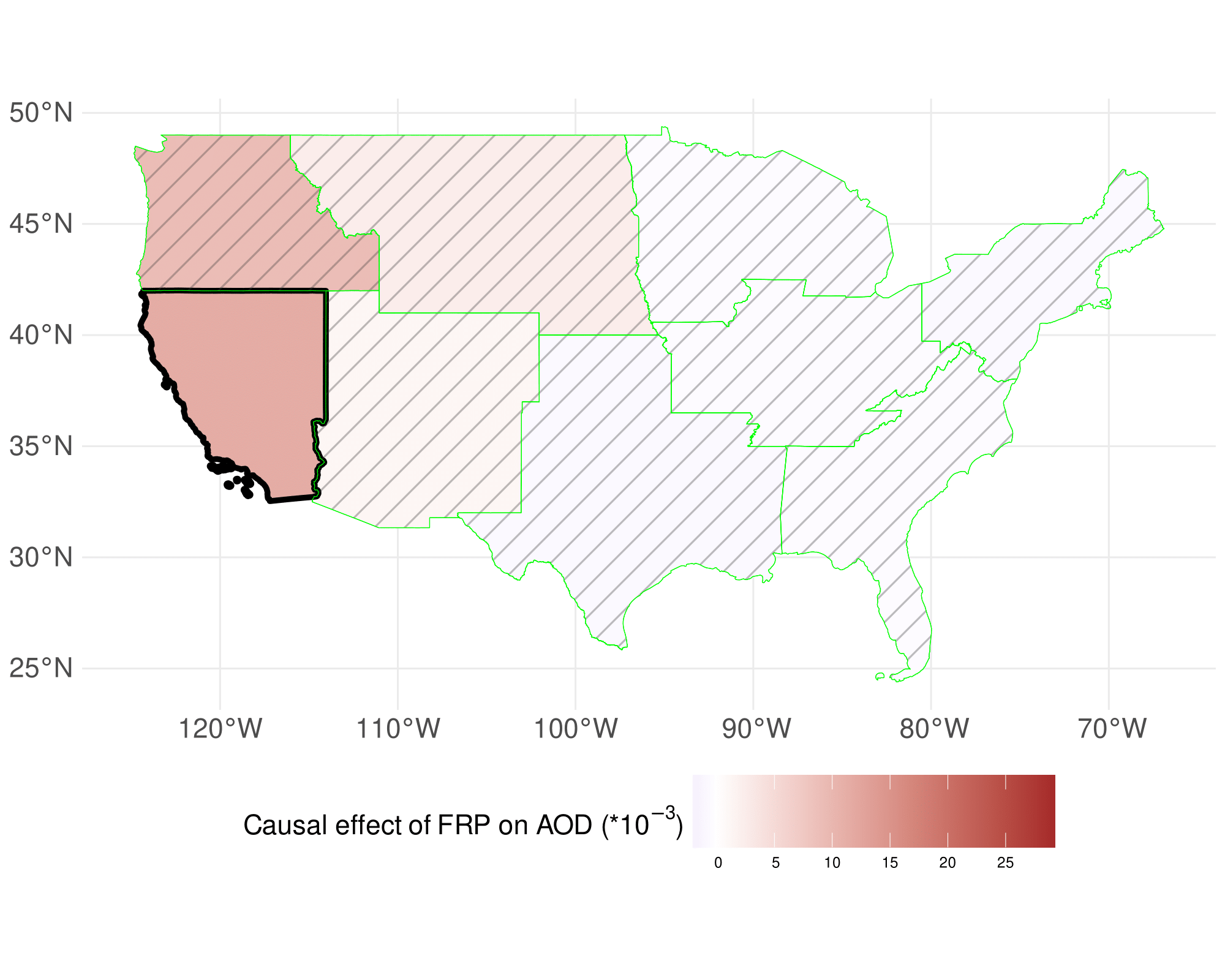}
        \vspace{-60pt}  
        \caption{Linear Regression}
        \label{fig:8-4}
    \end{subfigure}
    
    \vspace{-2pt}  
    \caption{Estimates of regional causal effect of FRP on AOD with the spatial $\tau$-QPE (QLSCM) for $\tau = 0.1$ (top left), $\tau = 0.5$ (top right), $\tau = 0.9$ (bottom left) and the ACE (LSCM; bottom right). The black and green outlines respectively indicate regions that exhibit significant and non-significant causal partial effects. Red regions represent positive causal partial effects, while blue regions represent negative causal partial effects.}
    \label{figure:44}
\end{figure}

\subsection{Daily State-wise Analysis}
\label{sec:state-wise}

\begin{hyphenrules}{nohyphenation}
We now investigate causal relationships across all states in the CONUS; see Figure~\ref{figure:11}. The state-level results reveal a more detailed pattern: Northwestern and Western states (California, Oregon, Idaho) exhibit statistically significant positive causal partial effects in the upper tail ($\tau = 0.9$) while being non-significant at lower quantiles ($\tau = 0.1, 0.5$).
This pattern indicates that wildfire intensity primarily affects air quality during extreme aerosol concentration events rather than typical conditions, which coincides with the previous regional analysis. The spatial $\tau$-QPE estimates reveal heterogeneous marginal effects missed by the ACE, providing crucial insights for understanding when and where wildfire smoke most severely impacts air quality.

The observed statistically significant negative effects in some northeastern states (shown in blue with black outline) may result from two factors: 1) spatial advection effects from nearby regions not captured in our local analysis, such as wildfire smoke from Quebec affecting the northeastern U.S. \citep{debell2004major}, and 2) different fire types contributing to distinct physical and chemical processes \citep{jaffe2020wildfire}. However, while statistically significant, these negative effects are much weaker than the positive effects and thus scientifically insignificant.
\end{hyphenrules}

\begin{figure}[t!]
    \centering
    \setlength{\intextsep}{0pt}
    \setlength{\abovecaptionskip}{1pt}
    \setlength{\belowcaptionskip}{0pt}

    \begin{subfigure}{0.4\textwidth}  
        \centering
        \includegraphics[width=\linewidth]{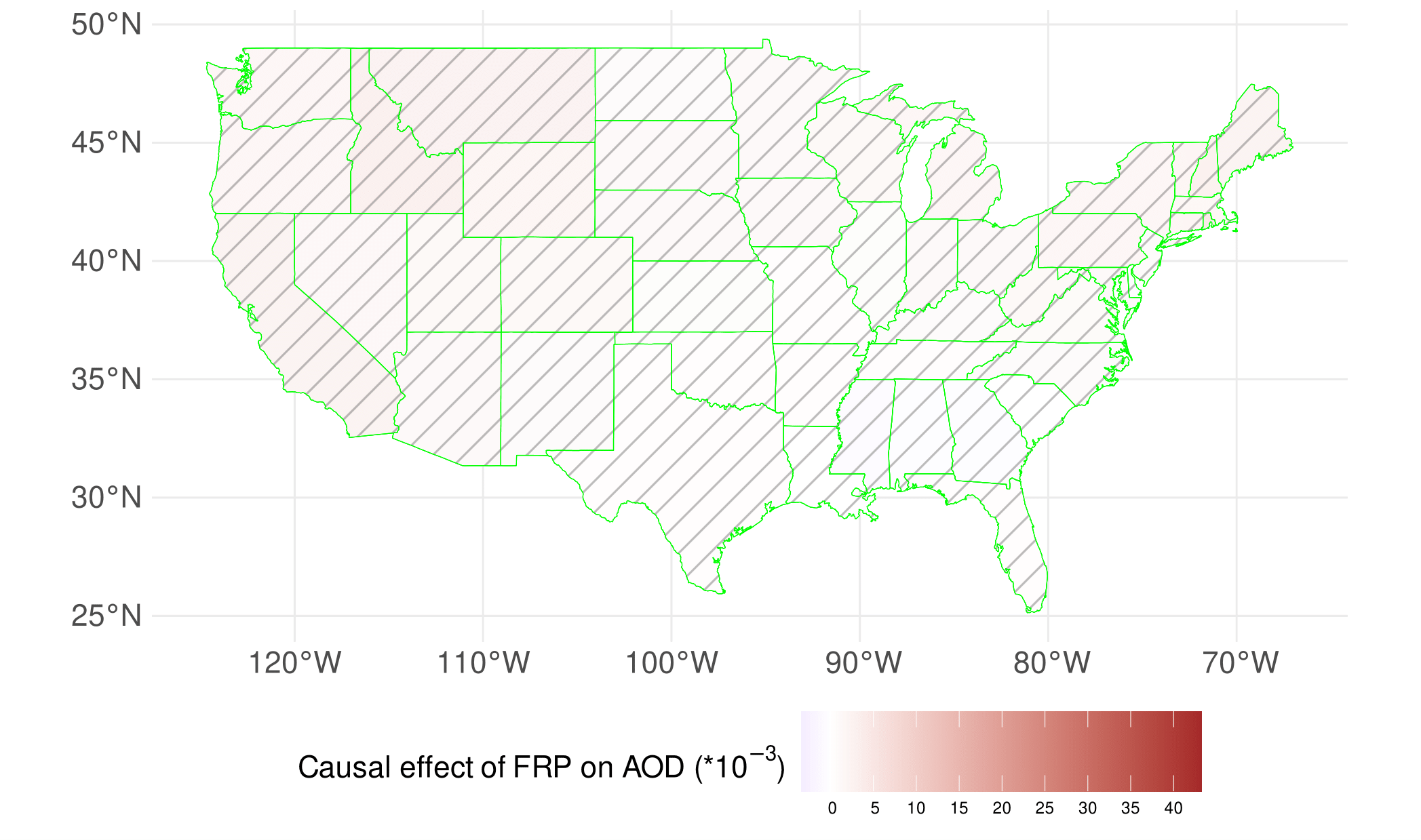}
        \vspace{-45pt} 
        \caption{$\tau = 0.1$}
        \label{fig:1-1}
    \end{subfigure}
    \hspace{1.5cm}
    \begin{subfigure}{0.4\textwidth}  
        \centering
        \includegraphics[width=\linewidth]{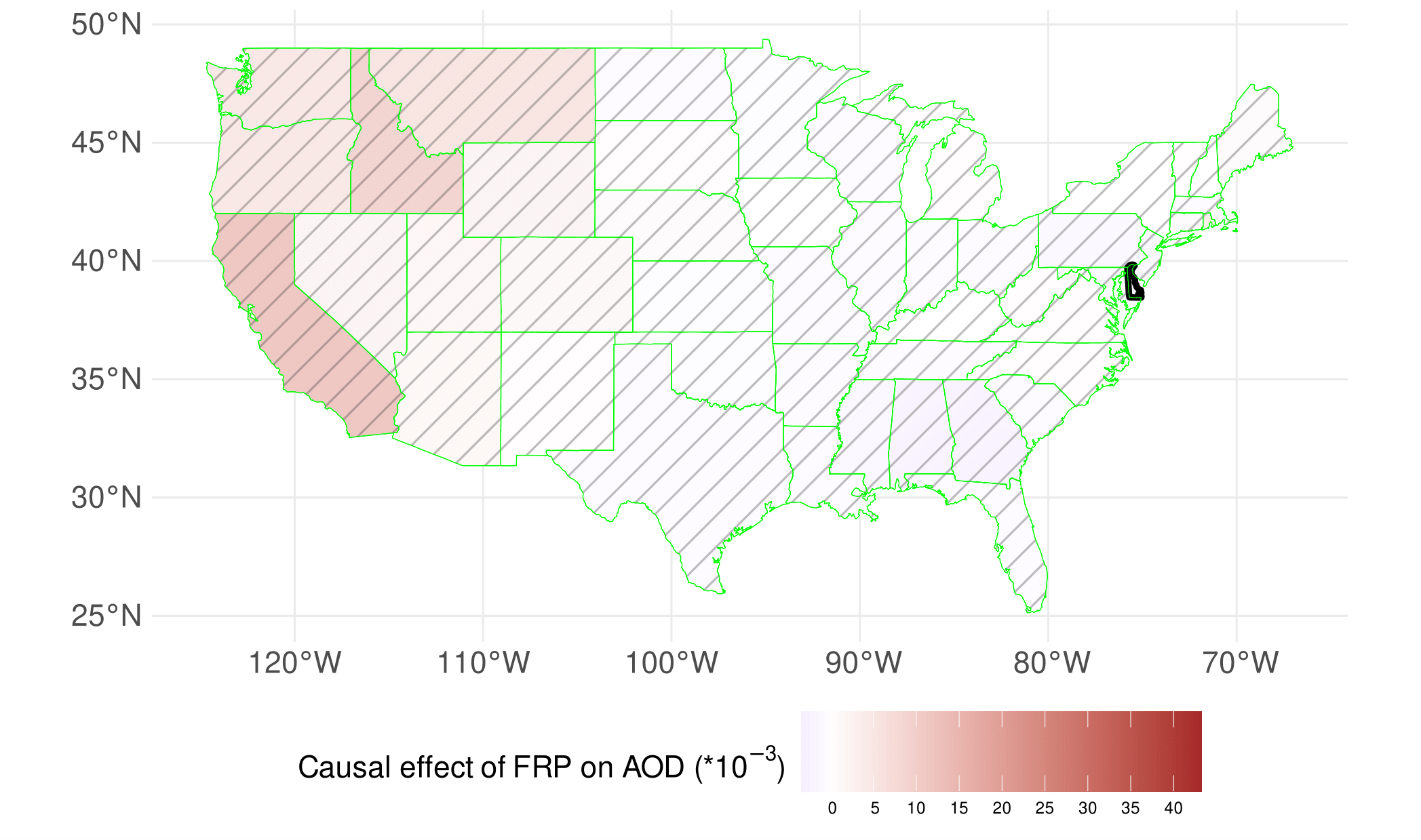}
        \vspace{-45pt} 
        \caption{$\tau = 0.5$}
        \label{fig:1-2}
    \end{subfigure}
    \vskip0.2\baselineskip  
    \begin{subfigure}{0.4\textwidth}  
        \centering
        \includegraphics[width=\linewidth]{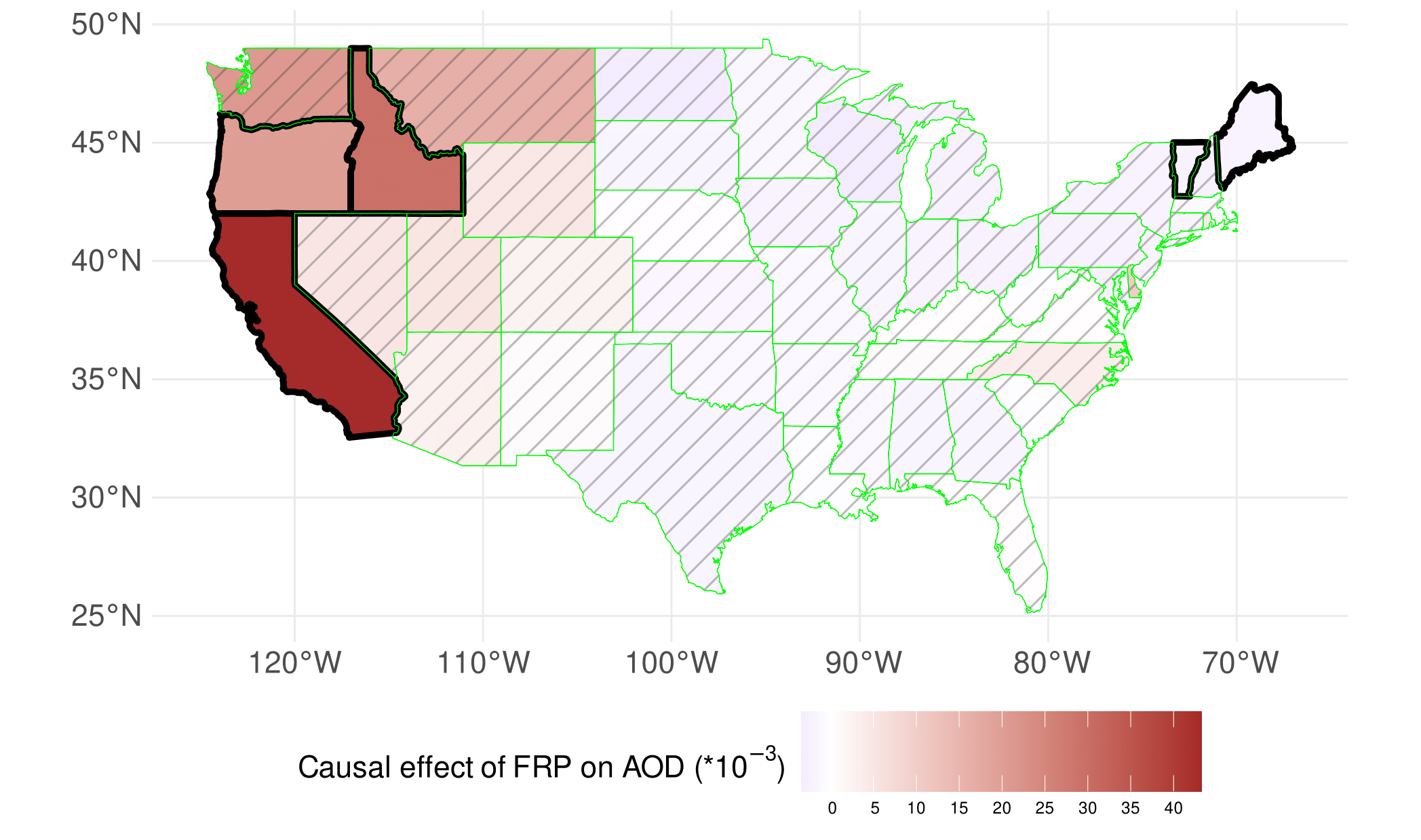}
        \vspace{-45pt} 
        \caption{$\tau = 0.9$}
        \label{fig:1-3}
    \end{subfigure}
    \hspace{1.5cm}
    \begin{subfigure}{0.4\textwidth}  
        \centering
        \includegraphics[width=\linewidth]{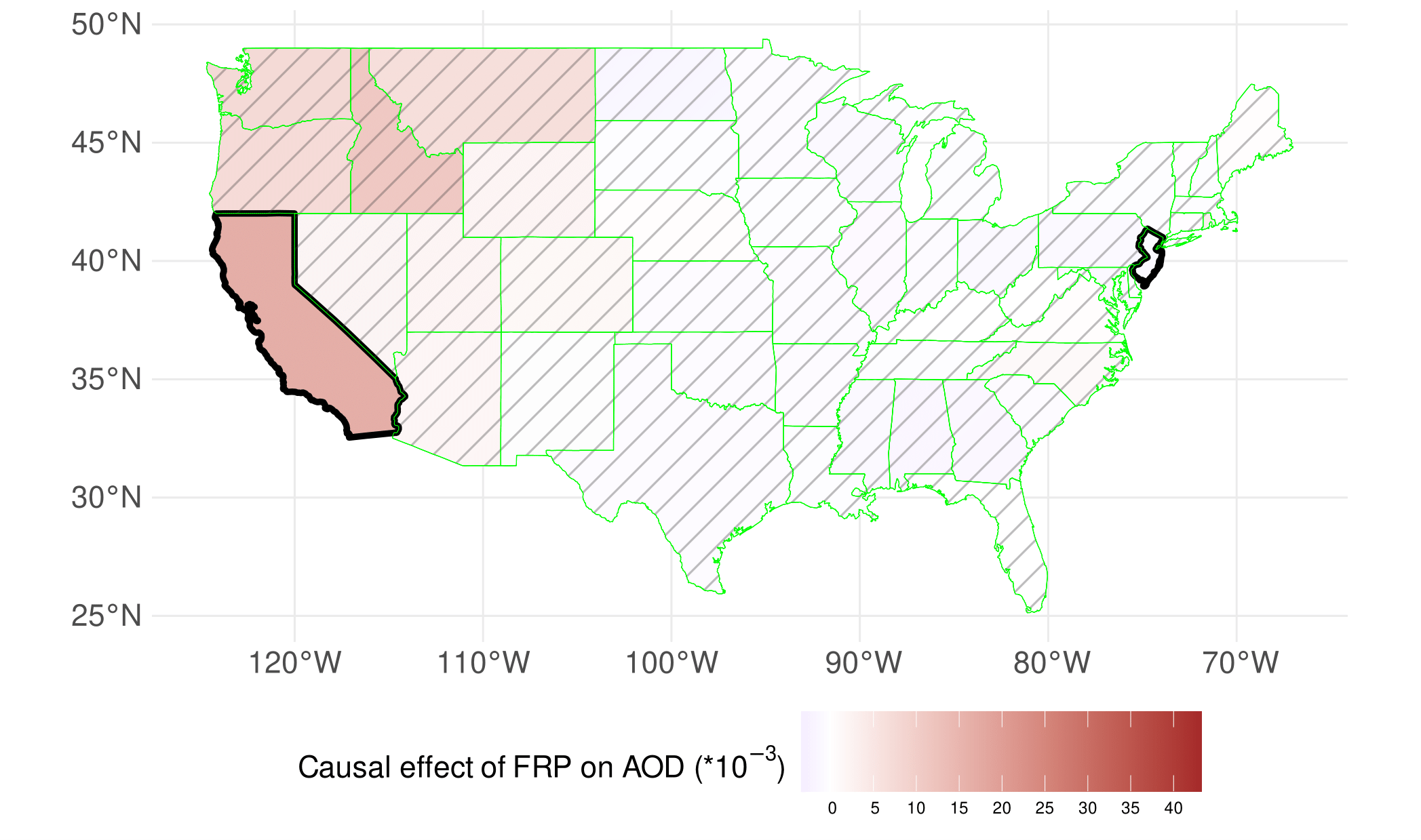}
        \vspace{-45pt} 
        \caption{Linear Regression}
        \label{fig:1-4}
    \end{subfigure}
    \vspace{-5pt}  
    \caption{Estimates of state-wise causal effect of FRP on AOD with spatial $\tau$-QPE (QLSCM) and ACE (LSCM) for $\tau = 0.1, 0.5, 0.9$. The black and green outlines respectively indicate regions that exhibit significant and non-significant causal partial effects. Red regions represent positive causal partial effects, while blue regions represent negative causal partial effects.}
    \label{figure:11}
\end{figure}
\subsection{Daily Analysis for Representative States}
\label{sec:rep_states}
\begin{hyphenrules}{nohyphenation}
We now focus on representative states chosen based on their wildfire exposure levels, as quantified by the Wildfire Risk Index (WRI) score and classification system \citep{NRI2024}, and their total AOD accumulation throughout the observation period. The WRI characterizes each community's wildfire susceptibility relative to other regions within the CONUS. The selected states are as follows: 1) Minnesota (MN), with relatively low wildfire risk and low AOD accumulation, 2) Oregon (OR), with relatively moderate wildfire risk and medium AOD accumulation, 3) California (CA), with relatively high wildfire risk and high AOD accumulation, and 4) Arizona (AZ), with relatively high wildfire risk but low AOD accumulation.

\begin{figure}[t!]
    \centering
    \includegraphics[width=0.65\linewidth]{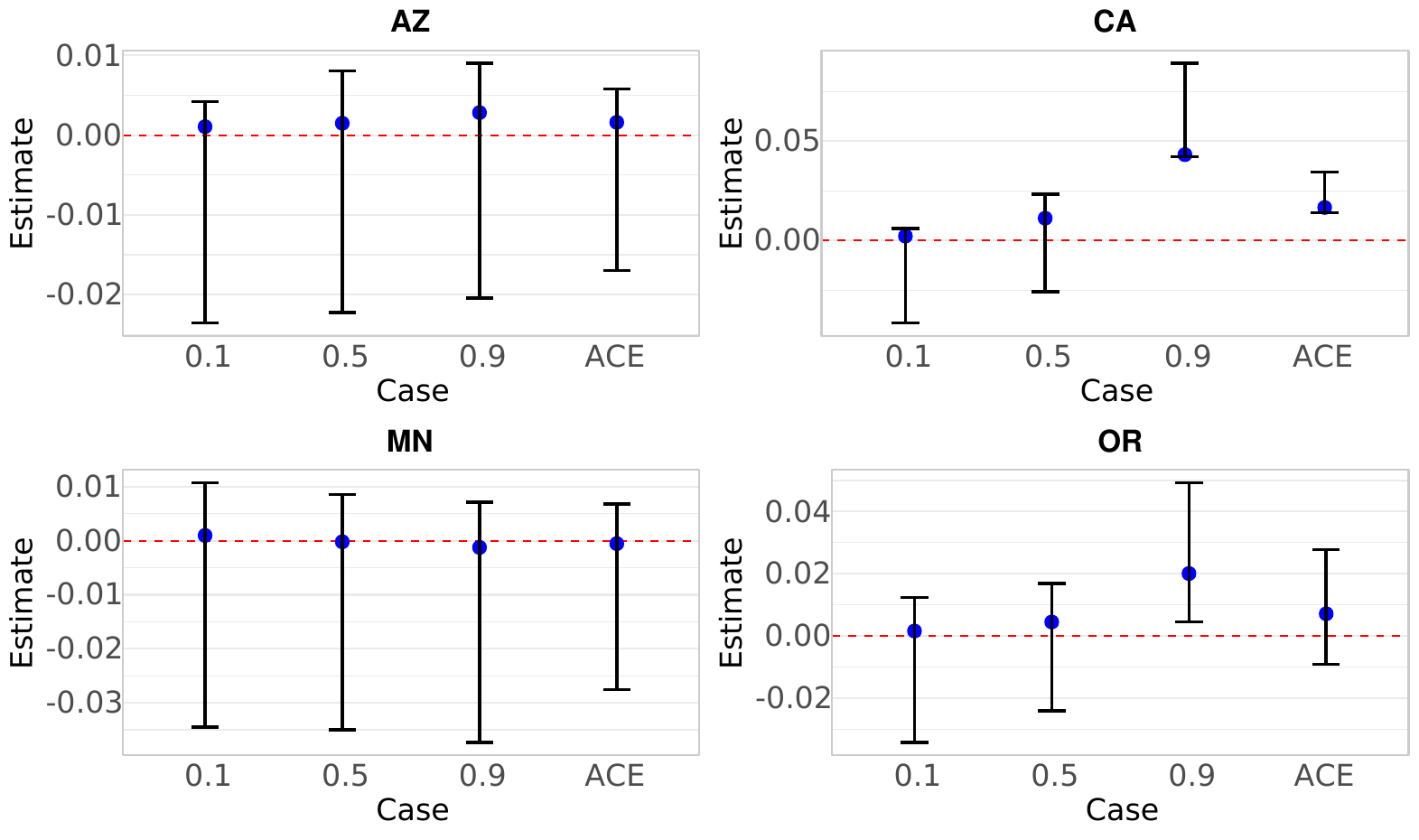}
    \caption{State-wise marginal causal effects of FRP on AOD based on the spatial $\tau$-QPE with $\tau = 0.1, 0.5, 0.9$ and the ACE (left to right) for Arizona (AZ), California (CA), Minnesota (MN), and Oregon (OR). Blue points are estimates, and the vertical segments are 99\% bootstrap confidence intervals. The red dashed line indicates zero effect.}
    \label{fig:daily_contemp}
\end{figure}

Figure~\ref{fig:daily_contemp} shows the spatial $\tau$-QPE and ACE estimates as well as corresponding 99\% confidence intervals. We observe distinct regional patterns in the causal relationship between wildfire intensity and air quality. CA and OR have a clear quantile-dependent effect, with non-significant effects at $\tau = 0.1, 0.5$ and significant positive effects at $\tau = 0.9$. Our results show that wildfire intensity predominantly affects air quality during extreme aerosol events in fire-prone western states (CA and OR), while AZ and MN exhibit negligible effects across all quantiles, with confidence intervals including zero. The ACE estimates fail to characterize these quantile-specific causal mechanisms, providing insufficient insight into the heterogeneous causal structure. Importantly, empirical evidence strongly favors QLSCM: for Oregon, QLSCM identifies statistically significant positive marginal causal effects, whereas LSCM yields insignificant estimates, highlighting the superiority of the quantile-based framework.
\end{hyphenrules}

\begin{hyphenrules}{nohyphenation}
To examine how the marginal causal effect of FRP on AOD varies across different quantile levels within California, we conduct a QLSCM analysis for $\tau = 0.1, 0.2, \ldots, 0.9$. \mbox{Figure} \ref{fig:CA_FRP9} reveals a clear heterogeneous relationship where the effect of FRP on AOD is not constant across the AOD distribution. For lower quantiles ($\tau \le 0.6$), the effect is statistically insignificant, 
but for larger quantiles ($\tau \ge 0.7$), the effect becomes statistically significant. The magnitude of this positive effect strengthens as the quantile level increases, peaking at $\tau=0.9$. This indicates that FRP has a substantially larger impact on AOD when AOD levels are already high, a critical detail not captured by the conditional mean, which only shows a modest positive ACE.
\end{hyphenrules}
\begin{figure}[t!]
    \centering
    \includegraphics[width=0.46\textwidth]{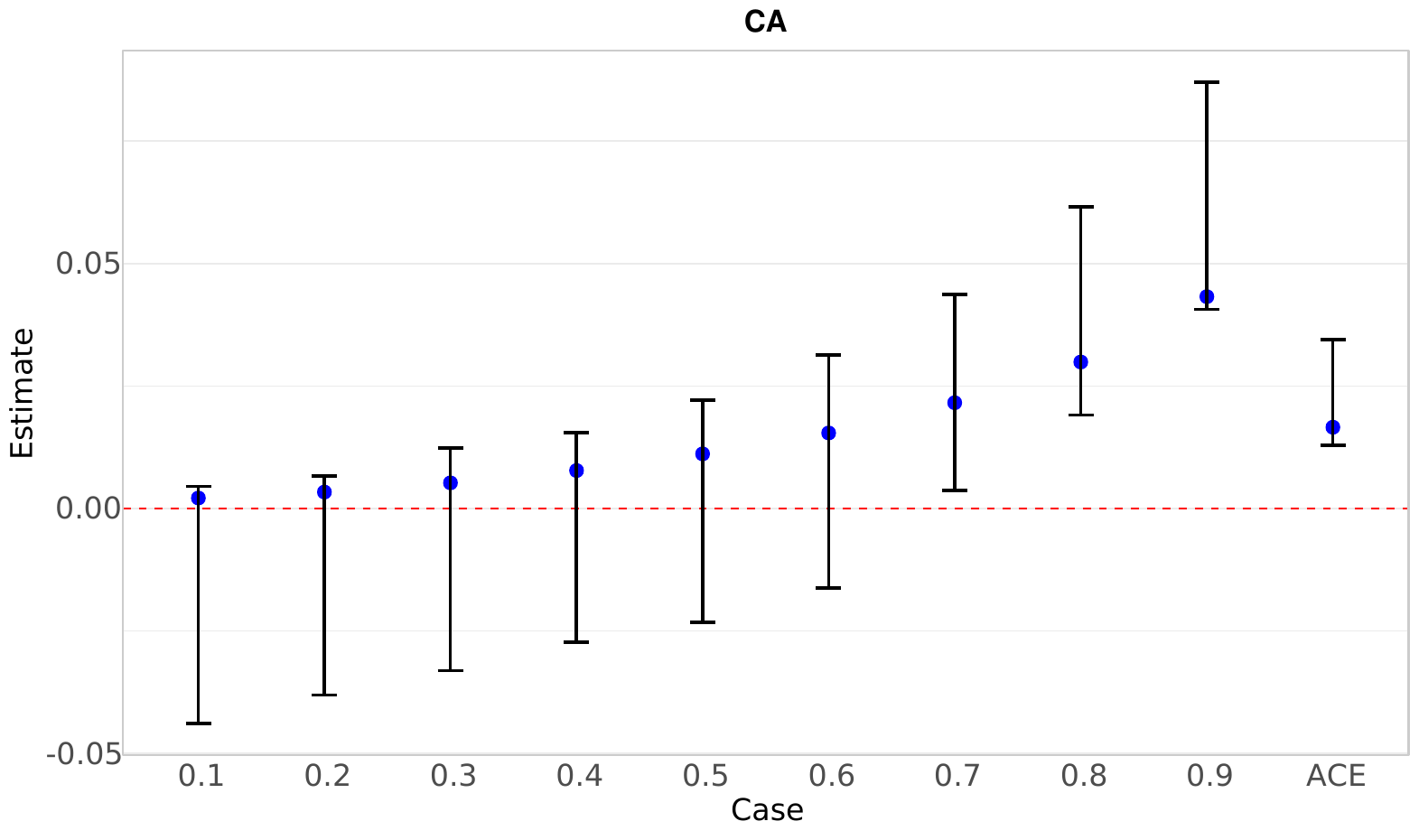}
    \caption{Estimates of spatial $\tau$-QPE of FRP on AOD in California with $\tau = 0.1, 0.2, \dots, 0.9$ and the ACE (left to right) for California (CA). Blue points are estimates, and vertical segments are 99\% bootstrap confidence intervals. The red dashed line indicates a zero effect.}
    \label{fig:CA_FRP9}
\end{figure}

\section{Discussion}
\label{sec:8}

\begin{hyphenrules}{nohyphenation}
In this paper, we introduced a quantile-based causal inference framework for spatio-temporal data. From a methodological standpoint, we proposed the QLSCM framework, which extends the LSCM approach beyond average causal effects to enable estimation of causal effects at specific quantiles of the outcome distribution. Our framework accommodates arbitrarily many observed and unobserved confounders, provided that unobserved confounders remain time-invariant. We established the asymptotic consistency of our estimator under assumptions that are less restrictive than those required by the LSCM framework, without imposing any parametric distributional requirements on the data-generating process. Additionally, we derived sufficient conditions guaranteeing consistent estimation of the spatial $\tau$-QPE at individual locations. Unlike methods that produce point estimates such as the conditional mean, the QLSCM framework can identify heterogeneous causal effects across the entire conditional distribution.

We applied the QLSCM framework to examine the causal relationship between wildfire intensity and air quality across the CONUS, utilizing daily observations during the peak fire season from 2003 to 2020. Regional analysis revealed that Western and Northwestern areas experience the most significant effects of wildfire intensity during extreme aerosol concentration episodes. Interestingly, Northeastern regions seem to exhibit an unexpected inverse relationship, where increased wildfire intensity appears associated with reduced AOD levels, though there negative effects are small and mostly insignificant. This counterintuitive pattern might be explained by the influence of wildfires originating in neighboring regions, particularly Canada, which are not explicitly incorporated into our localized analysis framework. State-level analysis revealed similar patterns in Oregon and California that support using our quantile-based approach. Wildfire intensity tends to significantly increase AOD only when AOD levels are already high, but shows no effect when AOD level is low or moderate. This suggests that wildfires matter most during already above-average air pollution. Results also reveal that this causal relationship differs across areas, indicating that the causal pathway between wildfire intensity and air quality may vary spatially. This suggests complex interactions between wildfire intensity and local environmental or meteorological conditions that deserve additional research. These results should be interpreted with caution, however, given our inability to adjust for hidden confounders that vary across both spatial and temporal dimensions. Additionally, achieving full control of confounders is impractical since they are difficult to quantify systematically, including fine-scale topographic variation and unobserved fire behavior characteristics~\citep{fernandes2003review}.

There are several areas of future work. First, the framework could incorporate tail index regression and other specialized regression models under an extremal setting~\citep[see, e.g.,][]{wang2009tail,Richards2024}. Second, future models should incorporate spatial advection effects~\citep{clarotto2024spde}—the movement of smoke and pollutants across regional boundaries—to provide more accurate causal estimates. Third, the quantile regression framework can be extended to characterize the entire outcome distribution, for instance through the recently proposed engression approach~\citep{shen2024engression}. Fourth, nonlinear quantile regression methods, e.g., B-splines, could also be exploited in the QLSCM framework.
\end{hyphenrules}

\begin{center}\textbf{Data Availability and Disclosure Statement}
\end{center}
\vspace{-0.3cm}
The data that support the findings of this study are available from the corresponding author, Z.G, upon reasonable request. The authors report there are no competing interests to declare.

\begin{center}\textbf{Ethics Declarations}
\end{center}
\vspace{-0.3cm}
During the preparation of this work, the author used Google Gemini (Pro, 2.0, 2.5 Pro) to assist with debugging the implementation code and paraphrasing portions of the text in order to ensure computational correctness and improve the clarity of the presentation.

\baselineskip=12pt
\bibliographystyle{agsm}
\bibliography{ref}

\newpage
\baselineskip=26pt
\begin{center}
{\large\bf Appendices}
\end{center}

\renewcommand{\theequation}{\thesection.\arabic{equation}}
\renewcommand{\thefigure}{\thesection.\arabic{figure}}
\renewcommand{\thetable}{\thesection.\arabic{table}}
\renewcommand{\thesection}{\Alph{section}}

\setcounter{figure}{0}
\setcounter{table}{0}
\setcounter{equation}{0}
\setcounter{theorem}{0}
\begin{appendix}





\section{Background on Causal Inference}
\label{app:bci}

\subsection{Definitions}
This compendium of fundamental definitions is incorporated to offer a theoretical groundwork for terminology used in our paper, primarily drawing from Section 6.1 of \citetapp{peters2017elements}.

\begin{definition}[graph terminology]\label{def:graph}
    Consider a random vector $\mathbf{X} = (X_1, \ldots, X_d)$ with index set $V := \{1, \ldots, d\}$, joint distribution $\PP_{\mathbf{X}}$, and density $p(\mathbf{x})$. A \textbf{graph} $\mathcal{G} = (V, E)$ consists of finitely many nodes $V$ and edges $E \subseteq V\times V$ with $(v,v) \not\in E$ for any $v \in \mathbf{V}$. We further have the following definitions:
    
    A node $i$ is called a \textbf{parent} of $j$ if $(i,j) \in E$ and $(j,i) \not\in E$ and a \textbf{child} if $(j,i) \in E$ and $(i,j) \not\in E$. The set of \textbf{parents} of $j$ is denoted by $\mathbf{PA}_j^{\mathcal{G}}$. Two nodes $i$ and $j$ are \textbf{adjacent} if either $(i,j) \in E$ or $(j,i) \in E$. We call $\mathcal{G}$ \textbf{fully connected} if all pairs of nodes are adjacent. We say that there is an \textbf{undirected} edge between two adjacent nodes $i$ and $j$ if $(i,j) \in E$ and $(j,i) \in E$. An edge between two adjacent nodes is \textbf{directed} if it is not undirected. We then write $i \rightarrow j$ for $(i,j) \in E$. We call graph $\mathcal{G}$ directed if all of its edges are directed.

    A \textbf{path} in $\mathcal{G}$ is a sequence of ($m\geq 2$) distinct vertices $i_1,\ldots,i_m$, such that there is an edge between $i_k$ and $i_{k+1}$ for all $k = 1,\ldots,m-1$. If $i_{k-1} \rightarrow i_k$ and $i_{k+1} \rightarrow i_k$, vertex $i_k$ is called a \textbf{collider relative to this path}. If $i_k \rightarrow i_{k+1}$ for all $k$, we speak of a \textbf{directed path} from $i_1$ to $i_m$ and call $i_1$ an \textbf{ancestor} of $i_m$ and $i_m$ a \textbf{descendant} of $i_1$.

    A graph $\mathcal{G}$ is called a \textbf{partially directed acyclic graph (PDAG)} if there is no directed cycle, that is, if there is no pair $(j, k)$ with directed paths from $j$ to $k$ and from $k$ to $j$. Grpah $\mathcal{G}$ is called a \textbf{directed acyclic graph (DAG)} if it is a PDAG and all edges are directed.
\end{definition}

Given the above graph terminologies, we introduce the definition of $d$-separation.

\begin{definition}[$d$-separation]\label{def:d-sep}
    A path between nodes $ i_1 $ and $ i_m $ is \textit{blocked by a set $\mathcal{S} \subset V$}, whenever there is a node $ i_k \in \mathcal{S} $ such that one of the following holds:
    \begin{enumerate}
        \item\, $ i_k\in \mathcal{S} $ and
        \begin{align*}
            i_{k-1}             & \to i_k \to i_{k+1}               \\
            \text{or}~~i_{k-1}  & \leftarrow i_k \leftarrow i_{k+1} \\
            \text{or}~~ i_{k-1} & \leftarrow i_k \to i_{k+1}
          \end{align*}
        \item \, neither $ i_k $ nor any of its descendants is in $ \mathcal{S} $ and
        \begin{align*}
            i_{k-1}\to i_k \leftarrow i_{k+1}, \;\;\;\text{e.g., $i_k$ is a collider relative to the path between $i_1$ and $i_m$.} 
          \end{align*}
      \end{enumerate}
    Two disjoint sets $ A,B\subset V $ are \textit{$d$-separated} by a disjoint set $ \mathcal{S} \subset V $ when each path between nodes in $ A $ and $ B $ is blocked by $ \mathcal{S} $. This is denoted as \(A\indep_{\mathcal{G}}B\mid S.\)
\end{definition}

\begin{definition}[global Markov property] 
Let $\mbX$ be a $d$-variate random vector and $\mathcal{G}=(V,E)$ some DAG with vertex set $V=[d]$ and edge set $E \subset V \times V$.
We say that $\mbX$ satisfies the global Markov property with respect to $\mathcal{G}$ when
\[A\indep_{\mathcal{G}}B\mid S \Rightarrow \mbX_A \indep \mbX_B \mid \mbX_S\]
for any disjoint subsets $A, B, S \subset V$.
The global Markov property connects $d$-separation with conditional independence, and we call any random vector $\mbX$ that satisfies the global Markov property with respect to $\mathcal{G}$ a directed graphical model with respect to $\mathcal{G}$.
\end{definition}

\begin{definition}[back-door criterion] A set of variables $\mbH$ satisfies the \textbf{back-door criterion} relative to a pair of variables $(\mbX, \mbY)$ in a DAG if:
\begin{enumerate}
    \item no variable in $\mbH$ is a descendant of $\mbX$.
    \item $\mbH$ blocks every path between $\mbX$ and $\mbY$.
\end{enumerate}
If a set $\mbH$ satisfies the backdoor criterion, the causal effect of $\mbX$ on $\mbY$ is identifiable.    
\end{definition}

\subsection{Explicit and Implicit Confounders}
\label{app:confounders}
We clarify the distinction between explicit and implicit forms of confounders. Let $X$ denote the exposure, $H$ the hidden confounder, and $Y$ the outcome. An explicit hidden confounder refers to cases where $H$ appears directly within the structural equation of $Y$, establishing causal pathways $H \rightarrow Y$ and $H \rightarrow X$. As an illustration, consider $Y = (1 + H) \cdot X + \varepsilon_Y$, with $\varepsilon_Y$ representing random error and $H$ explicitly entering the outcome's equation. We omit discussing how $H \rightarrow X$ since it is irrelevant. Conversely, an implicit hidden confounder describes situations where $H$ influences $Y$ through intermediate latent variables rather than appearing directly in the equation for $Y$. For example, consider $Y = V_1 + V_2 + \beta X$ with $V_1 \sim \mathcal{N}(\mu, H)$ and $V_2 \sim U(H, 2H)$, where $\mu,\beta \in \mathbb{R}$, $U(a, b)$ is the uniform distribution, and $V_1, V_2$ act as intermediate latent variables, establishing pathways $H \rightarrow V_i \rightarrow Y$ for $i = 1, 2$. The explicit form manifests as direct confounding bias via $H$'s direct participation in the outcome mechanism, whereas the implicit form exhibits confounding that operates through the parameters or distributions of intermediate variables.

\section{Background on Latent Spatial Confounder Model}
\label{app:lscm_defn}

\subsection{Classical Latent Spatial Confounder Model (LSCM)}
\label{subsec:lscm}

\begin{definition}[LSCM]
    Consider a spatio-temporal process $(\mb X, \mb Y, \mb H) := (\mbX_{\bs s}^t, Y_{\bs s}^t, \mbH_{\bs s}^t)_{({\bs s},t) \in \mathbb{R}^2 \times \mathbb{N}}$ over a univariate response $Y \in \mathbb{R}$, a vector of $d$ observed covariates $\mathbf{X} \in \mathbb{R}^d$, and a vector of $l$ hidden (unobserved) confounders $\mathbf{H} \in \mathbb{R}^\ell$. A causal graphical model over $(\mb X, \mb Y, \mb H)$ with causal structure $[\mb Y \mid \mb X, \mb H][\mb X \mid \mb H][\mb H]$ is an \textit{LSCM} if $\mb H$ is weak-sense stationary and time-invariant and there exists a measurable function $f : \mathbb{R}^{d+\ell+1} \rightarrow \mathbb{R}$ and an independent and identically distributed (IID) sequence $\bs\eps^1, \bs\eps^2, \dots$ of weak-sense stationary spatial error processes, independent of $(\mb X, \mb H)$, such that
    \begin{equation}
    Y_{\bs s}^t = f(\mbX_{\bs s}^t, \mbH_{\bs s}^t, \eps_{\bs s}^t) \quad \text{for all } ({\bs s}, t) \in \mathbb{R}^2 \times \mathbb{N}.
    \end{equation} Moreover, for all $\mathbf{x} \in \mathbb{R}^d$, $f({\mb x}, \mbH_{\bs 0}^1, \eps_{\bs 0}^1)$ must have finite expectation, i.e., $$\E_{\left(\mbH_{\bs 0}^1, \eps_{\bs 0}^1\right)}\left[f(\mbx, \mbH_{\bs 0}^1, \eps_{\bs 0}^1)\right] < \infty.$$
\end{definition}

\begin{remark}
\label{remark:1}
Given the assumptions of weak-sense stationarity and time-invariance of $\mb H$, and the IID conditions on $\bs\eps$, the terms $\mbH_{\bs s}^t$ and $\eps_{\bs s}^t$ are identical, respectively, across all \mbox{$(\bs s, t) \in \mathbb{R}^2 \times \mathbb{N}$}, which justifies the use of the notations $\mbH_{\bs 0}^1$ and $\eps_{\bs 0}^1$. Besides, as it is impractical to fully specify all variables in a spatio-temporal process, the LSCM framework focuses only on identifying specific causal relationships of interest, while remaining agnostic to other causal connections in the system.
    
The LSCM framework models the causal influence of exposure $\mathbf{X}$ on response variable $Y$ while leaving other causal relationships unspecified, rather than attempting to reconstruct the entire causal network. This targeted approach avoids the combinatorial complexity and estimation uncertainty inherent in methods that try to recover complete causal structures, such as maximum and partial ancestral graphs~\citepapp{zhang2008causal}.
\end{remark}

\subsection{Estimation of the Causal Estimand}
\label{subsec:estimation_classic}
The inferential target of the classical LSCM is the average causal effect (ACE) \(f_{\text{ACE}(\mb X \rightarrow Y)}(\mb x) := \E[f(\mb x, \mbH_{\bs 0}^1, \eps_{\bs 0}^1)].\) The justification of using the ACE for causal inference is given by Propositions~1 and 2 of \citetapp{christiansen_toward_2022}. Their proposed approach is to estimate \(f_{\text{ACE}(\mb X \rightarrow Y)}\) by first estimating the regression function \(f_{Y \mid (\mb X, \mb H)}: (\mb x, \mb h) \mapsto \E[Y_{\bs s}^t \mid \mbX_{\bs s}^t = \mb x, H_{\bs s}^t = \mb h]\), and then approximating the expectation:
\begin{equation}
    f_{\text{ACE}(\mb X \rightarrow Y)}(\mb x) := \E[f_{Y\mid (\mb X, \mb H)}(\mb x, \mbH_{\bs 0}^1)].
\end{equation}

The remaining challenge is to account for unobserved confounding variables. Assume the underlying dataset is $\left(\bs X_n^m, \bs Y_n^m\right) = (X_{\bs s}^t, Y_{\bs s}^t)_{(\bs s,t) \in \{\bs{s}_1, \dots, \bs{s}_n\} \times \{1, \dots, m\}}$ for $n$ spatial locations, $\{\bs s_1, \dots, \bs s_n\}$, and $m$ time points, $t \in 1, \dots, m$. Then, since \(\mb H\) is time-invariant, we can use the estimator
\begin{equation}
\hat{f}^{nm}_{\text{ACE}(\mb X \rightarrow Y)}(\bs X_n^m, \bs Y_n^m)(\mb x) := \frac{1}{n} \sum_{i=1}^{n} \hat{f}^m_{Y \mid \mb X}(\bs X^m_{\bs{s}_i}, \bs Y^m_{\bs{s}_i})(\mb x),\label{eq:ace_estimate}
\end{equation} where \(\hat{f}^m_{Y\mid X}(\bs X^m_{\bs s_i}, \bs Y^m_{\bs s_i})\) is any suitable estimator of \( x \mapsto f_{Y \mid (\mb X, \mb H)}(\mb x, \mb{h}_{\bs{s}_i}) \) computed from the observed data ($m$ replicates) at location \(\bs{s}_i\), and $\mb{h}_{\bs{s}_i}$ is the unobserved realization of $H_{\bs{s}_i}^1$.

\begin{remark}
    For every \( \bs s \in \{\bs{s}_1, \dots , \bs{s}_n\} \), we observe several time instances \( (\bs X^t_{\bs s}, Y^t_{\bs s}) \) for \( t \in \{1, . . . , m\} \), all with the same conditional distribution \( Y^t_{\bs s} \mid (\bs X^t_{\bs s}, \bs H^t_{\bs s}) \). Since \( \mb H \) is time-invariant, we can, for every \( \bs s \), estimate \( f_{Y\mid (\mb X, \mb H)}(\cdot, \mb h) \) for the realization \( \mb h_{\bs s} \) of \( \mbH_{\bs s}^1 \) using the data \( (\bs X^t_{\bs s}, Y^t_{\bs s}) \) for \( t \in \{1, . . . , m\} \), denoted by $\hat{f}^m_{Y \mid \mbX}(\bs X^m_{\bs s}, \bs Y^m_{\bs s})$. The ACE is then approximated by averaging estimates obtained from different spatial locations, which serves as a numerical approximation of the expectation. Under certain regularity assumptions, \citetapp{christiansen_toward_2022} prove that the estimator $\hat{f}^{nm}_{\text{ACE}(\mb X \rightarrow Y)}(\bs X_n^m, \bs Y_n^m)(\mb x)$ converges in probability to the true underlying ACE, $f_{\textrm{ACE}(\mb X \to Y)}(\mb x)$, for all $\mb x$.
\end{remark}

\section{Proofs}
\label{app:proof}
\subsection{Proof of Theorem~\ref{thm:1}}
\label{proof:1}

We first introduce a lemma before proving Theorem~\ref{thm:1}.
\begin{lemma}
If a quadruple $(\mb Y, \mb X, \mb H, \mb W)$ conforms to the QLSCM, we have that 
\begin{enumerate}
    \item[(i)] no node in $\mb H \cup \mb W$ is a descendant of $\mb X$,
    \item[(ii)] $\mb H \cup \mb W$ blocks every path between $\mb X$ and $\mb Y$ that contains an arrow into $\mb X$. 
\end{enumerate} In other words, $\mb H \cup \mb W$ satisfies the back-door criterion.
\end{lemma}

\begin{proof}
    The proof is trivial by noticing that $\mb H \cup \mb W$ blocks the path $\mb X \leftarrow \mb H\cup \mb W \rightarrow \mb Y$ and $\mbH \cup \mbW$ is already the root node so that there is no node belonging to the descendant set of $\mbX$.
\end{proof}

We now prove Theorem~\ref{thm:1}.
\par\noindent 
\textbf{Theorem 4.6} 
\textit{ Let $(\bs s, t) \in S \times \mathbb{N}$, where $S \subset \mathbb{R}^2$ is defined in Definition~\ref{def:qpe}, $\mb x \in \mathbb{R}^{d}$, and $\mbw \in \mathbb{R}^k$ be fixed, and consider any intervention on $\mb X$ s.t. $\mbX_{\bs s}^t = \mb x$ holds almost surely in the induced interventional distribution $\mathbb{P}_{\mb x}$. Denote the $\tau$-quantile of $Y_{\bs s}^t$ under $\PP_{\mb x}$ as $Q_{\tau}^{\PP_{\mb x}}(Y_{\bs s}^t)$ for a given $\bs s$ and $t$ and the spatially-integrated $\tau$-quantile of $\mb Y$ under $\PP_{\mb x}$ over space as \[Q_{\tau}^{\PP_{\mb x}}[\mb Y] := \frac{1}{\vert S\vert}\int_{S}Q_{\tau}^{\PP_{\mb x}}(Y_{\bs s}^t)\mathrm{d}\bs s.\] We then have that \(Q_{\tau}^{\PP_{\mb x}}[\mb Y] = g^{\tau}_{\textrm{QCE}(\mb X \to \mb Y)}(\mb x,\mbw),\) that is, $g^{\tau}_{\textrm{QCE}(\mb X \to \mb Y)}(\mb x, \mbw)$ is the $\tau$-quantile of $Y_{\bs s}^t$ under any intervention that enforces $\mbX_{\bs s}^t = \mb x$.}
\par\medskip 

\begin{proof}
We wish to estimate the spatial quantile causal effect. To do this, we first justify why we can define the $\tau$-quantile causal effect at location $s$ as $Q_{\tau}(Y \mid \mbX=\mbx, \mbH=\mbh_{\bs s}^1, \mbW = \mbw)$. To show this, we need first to show that under intervention $\mbX=\mbx$, we can identify the causality from statistical expressions. 

Using Theorem 2 from \citepapp{shpitser2012validity} and the previous Lemma, we know that 
\begin{align*}
\mathbb{P}(Y = y \mid do(\mathbf{X} = \mathbf{x})) &= \sum_{\mathbf{h},\mathbf{w}} \mathbb{P}(Y = y  \mid\mathbf{H} = \mathbf{h}, \mathbf{X} = \mathbf{x}, \mathbf{W} = \mathbf{w}) \mathbb{P}(\mathbf{H} = \mathbf{h}, \mathbf{W} = \mathbf{w}) \\
&= \mathbb{P}(Y = y \mid \mathbf{H} = \mathbf{h}_s^1, \mathbf{X} = \mathbf{x}, \mathbf{W} = \mathbf{w}) \mathbb{P}(\mathbf{H} = \mathbf{h}_s^1, \mathbf{W} = \mathbf{w}) \\
&= \mathbb{P}(Y = y \mid \mathbf{H} = \mathbf{h}_s^1, \mathbf{X} = \mathbf{x}, \mathbf{W} = \mathbf{w}) \cdot 1 \\
&= \mathbb{P}(Y = y \mid \mathbf{H} = \mathbf{h}_s^1, \mathbf{X} = \mathbf{x}, \mathbf{W} = \mathbf{w}) \quad \text{almost surely}
\end{align*} using 
$$\mathbb{P}(\mathbf{H} = \mathbf{h}, \mathbf{W} = \mathbf{w}') = \begin{cases}
1 & \text{if } \mathbf{h} = \mathbf{h}_s^1 \text{ and } \mathbf{w}' = \mathbf{w}, \\
0 & \text{otherwise}.
\end{cases}$$

Thus, defining the quantile-based estimands based on the aforementioned conditional distribution $\mathbf{Y} \mid(\mathbf{X},\mathbf{H},\mathbf{W})$ only needs $\inf\left\{m: \int_{-\infty}^m f(y \mid \mbx, \mbh_{\bs{s}}^1, \mbw) dy = \tau\right\}$, assuming the density of $\mathbf{Y} \mid (\mathbf{X}, \mathbf{H}, \mathbf{W})$ exists. This expression is equivalent to the conditional quantile $Q_{\tau}(\mathbf{Y} \mid \mathbf{X} =\mbx, \mathbf{H} = \mbh_{\bs{s}}^1, \mathbf{W} = \mbw)$.

Providing the justification of the definition of quantile causal effect, we now focus on the proof of equality. We know that
\begin{align*}
    Q_{\tau}^{\mathbb{P}_x}(\mathbf{Y}_{\bs{s}}^t) &= Q_{\tau}(\mathbf{Y}_{\bs{s}}^t \mid do(\mathbf{X}_{\bs{s}}^t = \mbx)), \quad \text{by definition}\\
    &= Q_{\tau}(\mathbf{Y}_{\bs{s}}^t \mid \mathbf{X}_{\bs{s}}^t = \mbx, \mathbf{H}_{\bs{s}}^t = \mbh_{\bs{s}}^t, \mathbf{W}_{\bs{s}}^t = \mbw), \quad \text{by the back-door criterion} \\
    &= Q_{\tau}(\mathbf{Y}_{\bs{s}}^t \mid \mathbf{X}_{\bs{s}}^t = \mbx, \mathbf{H}_{\bs{s}}^t = \mbh_{\bs{s}}^1, \mathbf{W}_{\bs{s}}^t = \mbw), \quad \text{by the time-invariance of $\mathbf{H}$}\\
    &= Q_{\tau}(f(\mbx, \mbh_{\bs{s}}^1, \mbw, \boldsymbol{\epsilon}_{\bs{s}}^t)), \quad \text{by the definition of the QLSCM} \\
    &= Q_{\tau}(f(\mbx, \mbh_{\bs{s}}^1, \mbw, \boldsymbol{\epsilon}_0^1)), \quad \text{as $\boldsymbol{\epsilon}^1, \boldsymbol{\epsilon}^2, \dots$ is an identical sequence of weak-sense stationary processes.}
\end{align*}
Thus, $Q_{\tau}^{\mathbb{P}_x}(\mathbf{Y}_{\bs{s}}^t)$ coincides with the integrand in Equation~\eqref{eq:def}, which finishes the proof.

\end{proof}

\subsection{Proof of Theorem~\ref{thm:2}}
\label{proof:2}
\par\noindent
\textbf{Theorem 4.7}
\textit{Let $g^{\tau}_{Y \mid (\mb X, \mb H, \mb W)}$ denote the quantile regression function $(\mb x, \mb h, \mb w) \mapsto Q_{\tau}(Y_{\bs s}^t \mid \mbX_{\bs s}^t = \mb x, \mbH_{\bs s}^t = \mb h, \mbW_{\bs s}^t = \mb w)$. We have that, for all $\mb x \in \bb{R}^{d}$ and $\mb w \in \bb{R}^k$, \begin{equation*} g^{\tau}_{\textrm{QPE}(\mb X \to \mb Y)}(\mb x,\mbw) = \frac{1}{\vert S\vert }\int_{S} \frac{\partial g^{\tau}_{Y \mid (\mb X, \mb H, \mb W)}(\mb x, \mb h_{\bs s}^1, \mb w)}{\partial \mb x} \mathrm{d}\bs s. \end{equation*}}
\par\medskip
\begin{proof}
\begin{hyphenrules}{nohyphenation}
The proof follows directly from the definition of the spatial $\tau$-QPE. Since the spatial $\tau$-QPE is formulated as a spatial average, the assumption of ergodicity ensures that the average over the full observational data—comprising a single observation at each location—converges to the underlying marginal causal effect as the spatial domain approaches infinity. Consequently, the identification result holds as an immediate corollary of Theorem~\ref{thm:1}.
\end{hyphenrules}
\end{proof}

\subsection{Proof of Theorem~\ref{thm:3}}
\label{proof:3}
\begin{lemma}
\label{lemma:2}
    Given any function $f(\bs s) \in \mathcal{L}_1$, i.e., absolute integrable function, with $\bs s \in S \subseteq \mathbb{R}^2$ and $S$ a rectangular domain, the following equality holds:
    \[\lim_{n \to \infty}\frac{1}{\sum_{i=1}^n\vert S_i\vert} \sum_{i=1}^n \vert S_i\vert f(\bs s_i) = \frac{1}{\vert S\vert} \int_{S} f(\bs s)\mathrm{d}\bs s,\] where $\bs s_i \in S_i$, and $(S_i)_{i=1,\dots, n}$ are any partitions of the space $S$ that satisfy the condition that all subrectangles (grid box, in our case) are non-overlapping and their union is $S$, with $\vert S_i\vert\to 0$ when $n \to \infty, \forall i,$ given the limit exists.
\end{lemma}
\begin{proof}
    The equality holds directly from the definition of the Riemann integral by Riemann sum, where the partition of the space is in accordance with the Cartesian coordinate system.
\end{proof}
We now prove Theorem~\ref{thm:3}.
\par\noindent
\textbf{Theorem 4.11} 
\textit{Let \((\mb Y, \mb X, \mb H, \mb W)\) follow the QLSCM defined in Definition~\ref{def:qlscm}. Let \((\bs s_n)_{n \in \mathbb{N}}\) be a sequence of spatial coordinates such that the marginalized process \((H_{\bs s_n}^1)_{n \in \mathbb{N}}\) has realizations $(\mb h_{\bs s_n}^1)_{n \in \mathbb{N}}$, and the spatial coordinates $\bs s_i$ coincide with the centroids of the gridded blocks $S_i$, where $\bigcup_i S_i = S$. Let \(\hat{g}^{\tau}_{Y \mid (\mb X, \mb W)} = (\hat{g}^{m_i, \tau}_{Y\mid (\mb X, \mb W)})_{m_i \in \mathbb{N}}\) be an estimator satisfying Assumption~\ref{ass:1}. We then have the following consistency result: for all \(\mbx \in \mathcal{X}, \mbw \in \mathcal{W}\), \(\delta > 0\), and \(\alpha > 0\), there exists \(N \in \mathbb{N}\) such that, for all \(n \geq N\), there exists \(M_n \in \mathbb{N}\) such that for all \(\min_{j \in \{1,\dots,n\}}{m_j} \geq M_n\) we have that \(\PP\left(\left\|\hat{g}^{nm, \tau}_{\textrm{QPE}(\mb X \to \mb Y)}(\bs{X}_n^{m_n}, \bs{Y}_n^{m_n}, \bs{W}_n^{m_n})(\mb x,\mbw) - g^{\tau}_{\textrm{QPE}(\mb X \to \mb Y)}(\mb x,\mbw)\right\| > \delta \right) \leq \alpha.\)}
\par\medskip

\begin{proof}
    Consider fixed \( \mbx \in \mathcal{X}, \mbw \in \mathcal{W}\) and $\bs s_i \in S_i$, where $\bigcup_i S_i = S$. For every \( n,m_i \in \mathbb{N} \), we have that
\begin{align*}
    & v(\bs X^{m_n}_n, \bs Y^{m_n}_n, \bs W^{m_n}_n) := \qest - \qestimand \\ &= \frac{1}{\sum_{i=1}^n \vert S_i\vert}\sum_{i=1}^n \vert S_i\vert \frac{\partial \hat{g}^{m_i, \tau}_{Y \mid (\mbX,\mbW)}(\bs X_{\bs s_i}^{m_i}, \bs Y_{\bs s_i}^{m_i}, \bs W_{\bs s_i}^{m_i})(\mbx, \mbw)}{\partial \mbx} - \frac{1}{\vert S\vert} \int_{S} \frac{\partial g^{\tau}_{Y \mid (\mbX,\mbH,\mbW)}(\mbx, \mbh_{\bs s}^1, \mbw)}{\partial \mbx}\mathrm{d}\bs s, \\
    &= \frac{1}{\sum_{i=1}^n \vert S_i\vert} \sum_{i=1}^{n} \vert S_i\vert \frac{\partial \left( \hat{g}^{m_i, \tau}_{Y\mid (\mbX,\mbW)}(\bs X_{\bs s_i}^{m_i}, \bs Y_{\bs s_i}^{m_i}, \bs W_{\bs s_i}^{m_i})(\mbx,\mbw) - g^{\tau}_{Y \mid (\mbX,\mbH,\mbW)}(\mbx, \mbh^1_{\bs s_i}, \mbw) \right)}{\partial \mbx}, \\
    &\quad\quad\quad \text{by Lemma \ref{lemma:2} as $\{S_i\}$ is a valid division of $S$.}
\end{align*}

Let $\alpha > 0$ be arbitrary. By Assumption~\ref{ass:1}, there exists $N \in \bb{N}$, we can for any such \( n > N \) there exists \( M_n \in \mathbb{N} \), such that for all \( i = 1,\ldots,n \) and all \( m_i > M_n \), it holds that 
\[
\mathbb{P}\left(\left\|\frac{\partial \hat{g}^{m_i, \tau}_{Y\mid (\mbX,\mbW)}(\bs X_{\bs s_i}^{m_i}, \bs Y_{\bs s_i}^{m_i}, \bs W_{\bs s_i}^{m_i})(\mbx, \mbw)}{\partial \mbx} - \frac{\partial g^{\tau}_{Y\mid (\mbX,\mbH,\mbW)}(\mbx, \mbh^1_{\bs s_i},\mbw)}{\partial \mbx} \right\| > \delta \right) \leq \alpha / n.
\] Hence, for any $\min_i m_i > M_n$, we have that
\begin{align*}
   & \mathbb{P}\left( \left\| v(\bs X^{m_n}_n, \bs Y^{m_n}_n, \bs W^{m_n}_n) \right\| > \delta \right) \\ &\leq \PP \left( \frac{1}{\sum_{i=1}^n \vert S_i\vert} \sum_{i=1}^{n} \vert S_i\vert \left\| \frac{\partial \hat{g}^{m_i, \tau}_{Y\mid (\mbX,\mbW)}(\bs X_{\bs s_i}^{m_i}, \bs Y_{\bs s_i}^{m_i}, \bs W_{\bs s_i}^{m_i})(\mbx,\mbw)}{\partial \mbx} - \frac{\partial g^{\tau}_{Y \mid (\mbX,\mbH,\mbW)}(\mbx, \mbh^1_{\bs s_i},\mbw)}{\partial \mbx} \right\| > \delta \right) \\
    &\leq \PP \left( \bigcup_{i=1}^n \left\{\frac{\vert S_i\vert }{\sum_{i=1}^n \vert S_i\vert}\left\| \frac{\partial \hat{g}^{m_i, \tau}_{Y \mid (\mbX,\mbW)}(\bs X_{\bs s_i}^{m_i}, \bs Y_{\bs s_i}^{m_i}, \bs W_{\bs s_i}^{m_i})(\mbx,\mbw)}{\partial \mbx} - \frac{\partial g^{\tau}_{Y \mid (\mbX,\mbH,\mbW)}(\mbx, \mbh^1_{\bs s_i},\mbw)}{\partial \mbx} \right\| > \frac{\delta \vert S_i\vert}{\sum_{i=1}^n \vert S_i\vert}\right\} \right), \\ 
&\quad\quad\quad\quad\quad\quad\quad\quad\quad\quad\text{by the triangular inequality}\\
    &= \PP \left( \bigcup_{i=1}^n \left\{\left\| \frac{\partial \hat{g}^{m_i, \tau}_{Y \mid (\mbX,\mbW)}(\bs X_{\bs s_i}^{m_i}, \bs Y_{\bs s_i}^{m_i}, \bs W_{\bs s_i}^{m_i})(\mbx,\mbw)}{\partial \mbx} - \frac{\partial g^{\tau}_{Y \mid (\mbX,\mbH,\mbW)}(\mbx, \mbh^1_{\bs s_i},\mbw)}{\partial \mbx} \right\| > \delta \right\} \right) \\
    &\leq \sum_{i=1}^{n} \PP \left( \left\| \frac{\partial \hat{g}^{m_i, \tau}_{Y \mid (\mbX,\mbW)}(\bs X_{\bs s_i}^{m_i}, \bs Y_{\bs s_i}^{m_i}, \bs W_{\bs s_i}^{m_i})(\mbx,\mbw)}{\partial \mbx} - \frac{\partial g^{\tau}_{Y \mid (\mbX,\mbH,\mbW)}(\mbx, \mbh^1_{\bs s_i},\mbw)}{\partial \mbx}\right\| > \delta \right) \\
    &\leq \sum_{i=1}^{n} \frac{\alpha}{n} = \alpha,
\end{align*} and hence the statement has been proved.
\end{proof}

\subsection{Proof of Theorem~\ref{thm:4}}
\label{proof:4}

\begin{definition}
A time series $\{u_t\}_{t \in \mathbb{Z}}$ is $\alpha$-mixing if:
\begin{equation}
\label{eq:mixing}
\alpha(n) = \sup_{A \in \mathcal{F}_{-\infty}^t, B \in \mathcal{F}_{t+n}^{\infty}} |\PP(A \cap B) - \PP(A)\PP(B)| \rightarrow 0 \quad \text{as} \quad n \rightarrow \infty
\end{equation}
where $\mathcal{F}_j^k$ is the $\sigma$-algebra generated by $\{u_t : j \leq t \leq k\}$.
\end{definition}

We wish to prove the sufficient conditions to the Assumption~\ref{ass:1} which is as
follows:

\setcounter{assumption}{0}
\begin{assumption}
    There exists $\mathcal{X} \subset \mathbb{R}^d$, \textit{s.t.}
for all $\mathbf{x \in} \mathcal{X}, \mathbf{w} \in  \mathcal{W} \text{ and } \bs{s}_i \in S
\subseteq \mathbb{R}^2$, it holds that
\[ \frac{\partial \hat{g}^{m_i, \tau}_{Y \mid (\mathbf{X},\mbW)}
   (\bs{X}_{\bs{s}_i}^{m_i}, \bs{Y}_{\bs{s}_i}^{m_i}, \bs{W}_{\bs{s}_i}^{m_i}) (\mathbf{\mathbf{x}}, \mbw)}{\partial
   \mathbf{x}} - \dfrac{\partial g_{Y \mid (\mathbf{\mathbf{X}}, \mathbf{H}, \mbW)}^{\tau} (\mathbf{x}, \mathbf{h}^1_{\bs{s}_i},\mbw)}{\partial
   \mathbf{x}}  \; \overset{\PP}{\rightarrow} 0 \] as $m_i \rightarrow \infty.$
\end{assumption}

\begin{theorem}[Sufficient conditions for Assumption~\ref{ass:1}]
\label{thm:4}
Suppose there exist functions $f_{11}: \mathbb{R}^{d} \times \mathbf{\Omega} \rightarrow \mathbb{R}$, $f_{12}: \mathbb{R}^{\ell} \times \mathbb{R}^{k} \rightarrow \mathbb{R}$, and $f_2 : \mathbb{R}^{\ell + k + 1} \rightarrow \mathbb{R}$ such that the QLSCM admits the multiplicatively separable form
\begin{align}
  Y_{\bs{s}}^t & = g(\mathbf{X}_{\bs{s}}^t, \mathbf{H}_{\bs{s}}^t, \mathbf{W}_{\bs{s}}^t, \varepsilon_{\bs{s}}^t) \nonumber\\
  & = f_{11} (\mathbf{X}^t_{\bs{s}}, \bs\beta) \cdot f_{12} (\mathbf{H}_{\bs{s}}^t, \mathbf{W}_{\bs{s}}^t) + f_2 (\mathbf{H}_{\bs{s}}^t, \mathbf{W}_{\bs{s}}^t, \varepsilon_{\bs{s}}^t) \label{eq:qlscm_form}\\
  & = f_{11} (\mathbf{X}^t_{\bs{s}}, \bs\beta) \cdot f_{12} (\mathbf{H}_{\bs{s}}^t, \mathbf{W}_{\bs{s}}^t) + \vartheta_{\bs s}^t,
\end{align}
for all $(\bs{s}, t)$, where $\vartheta_{\bs s}^t := f_2 (\mathbf{H}_{\bs{s}}^t, \mathbf{W}_{\bs{s}}^t, \varepsilon_{\bs{s}}^t)$. Here, $\mathbf{\Omega}$ is an open subset of $\mathbb{R}^d$ that contains the unknown $d$-dimensional parameter vector $\bs\beta$. Note that $\bs\beta$ does not include an intercept term.

Following the notation in Assumption~\ref{ass:1}, for each location $\bs{s}_i \in S$, there are $m_i$ samples representing $m_i$ time instances that form a time series. For all $\bs{s}_i \in S$ and all $t \in \{1,\dots,m_i\}$, we impose the following assumptions on the error processes (we omit the location subscript for notational simplicity):

\textbf{Assumption C.1.} For the distribution $F_t(z)$ of $\vartheta_t$, we have $F_t(0) = P(\vartheta_t \leq 0) = \tau$, where $0 < \tau < 1$ for all $t$.

\textbf{Assumption C.2.} The error process $\{\vartheta_t\}$ is $\alpha$-mixing as defined in \eqref{eq:mixing}.

\textbf{Assumption C.3.} There exist positive constants $a_0$ and $b_0$ such that for all $|z| \leq b_0$ and all $t$, 
\[
\min[F_t(|z|) - F_t(0), F_t(0) - F_t(-|z|)] \geq a_0|z|.
\]

Additionally, we require the following assumptions on the covariates and regression function:

\textbf{Assumption C.4.} The function $f_{11}(\mathbf{x}, \bs\beta)$ is linear in $\mathbf{x}$ for all $\bs\beta \in \mathbf{\Omega}$. That is, $f_{11}(\mathbf{x}, \bs\beta) = \bs\beta^{\top} \mathbf{x}$.

\textbf{Assumption C.5.} (Identifiable Uniqueness) For every $\epsilon > 0$, there exists a positive constant $\delta_\epsilon$ such that for all $\bs\beta \in \bs\Omega$,
\begin{equation*}
\liminf_{m_i \to \infty} \inf_{\|\tilde{\bs\beta}-\bs\beta\| \geq \epsilon} {m_i}^{-1}\sum_{t=1}^{m_i}\left|f_{11}( \mathbf{x}_t, \tilde{\bs\beta}) f_{12}(\mathbf{h}_t, \mathbf{w}_t) - f_{11}( \mathbf{x}_t, \bs\beta) f_{12}(\mathbf{h}_t, \mathbf{w}_t)\right| > \delta_\epsilon.
\end{equation*}

\textbf{Assumption C.6.} (Moment Boundedness) For some $\epsilon > 0$ and all $\bs\beta \in \bs\Omega$,
\begin{equation*}
\limsup_{m_i \to \infty} {m_i}^{-1}\sum_{t=1}^{m_i}\left|f_{11}( \mathbf{x}_t, \bs\beta) f_{12}(\mathbf{h}_t, \mathbf{w}_t)\right|^{1+\epsilon} < \infty.
\end{equation*}

\textbf{Assumption C.7.} For every $\bs\beta \in \bs\Omega$ and every $\epsilon > 0$, there exists a $\delta > 0$ such that
\begin{equation*}
\limsup_{m_i \to \infty} \sup_{\|\tilde{\bs\beta}-\bs\beta\| \leq \delta} {m_i}^{-1}\sum_{t=1}^{m_i}\left|f_{11}( \mathbf{x}_t, \tilde{\bs\beta}) f_{12}(\mathbf{h}_t, \mathbf{w}_t) - f_{11}( \mathbf{x}_t, \bs\beta) f_{12}(\mathbf{h}_t, \mathbf{w}_t)\right| < \epsilon.
\end{equation*}

Under Assumptions C.1--C.7, Assumption~\ref{ass:1} is satisfied.
\end{theorem}

\begin{proof} 
By the time-invariance of $\mathbf{H}_{\bs s}^t$ and our assumption that we observe only a single replicate at each time instance, the probability statement $\PP(\mbH_{\bs{s}}^1 = \cdots = \mbH_{\bs{s}}^{m_i}) = 1$ implies that the realizations of the hidden confounder processes satisfy $\mb h_{\bs{s}}^1 = \cdots = \mb h_{\bs{s}}^{m_i} = \mb h_{\bs{s}}$, which represents a fixed (but unknown) value at location $\bs{s}$.

\textbf{Step 1: Establishing the conditional quantile function.}
Under Assumption C.1, we have $F_t(0) = P(\vartheta_t \leq 0) = \tau$, which implies that the $\tau$-th quantile of the error term satisfies $Q_\tau(\vartheta_{\bs{s}}^t \mid \mathbf{H}_{\bs{s}}^t, \mathbf{W}_{\bs{s}}^t) = 0$. From the structural form in \eqref{eq:qlscm_form}, the $\tau$-th conditional quantile function of $Y_{\bs{s}}^t$ given $(\mathbf{X}_{\bs{s}}^t, \mathbf{H}_{\bs{s}}^t, \mathbf{W}_{\bs{s}}^t)$ at evaluation point $(\mathbf{x}, \mathbf{h}_{\bs{s}}, \mathbf{w})$ is
\begin{equation*}
g_{Y \mid (\mathbf{X}, \mathbf{H}, \mbW)}^{\tau}(\mathbf{x}, \mathbf{h}_{\bs{s}}, \mathbf{w}) = Q_\tau(Y_{\bs{s}}^t \mid \mathbf{X}_{\bs{s}}^t = \mathbf{x}, \mathbf{H}_{\bs{s}}^t = \mathbf{h}_{\bs{s}}, \mathbf{W}_{\bs{s}}^t = \mathbf{w}) = f_{11}(\mathbf{x}, \bs\beta) \cdot f_{12}(\mathbf{h}_{\bs{s}}, \mathbf{w}).
\end{equation*}

Since $\mathbf{H}_{\bs{s}}^t \equiv \mb h_{\bs{s}}$ is constant at location $\bs{s}$, for any fixed evaluation point $\mathbf{w}$, the quantity $f_{12}(\mb h_{\bs{s}}, \mathbf{w})$ is a constant (though unknown due to the latent nature of $\mathbf{h}_{\bs{s}}$).

\textbf{Step 2: Computing partial derivatives via Assumption C.4.}
Under Assumption C.4, $f_{11}(\mathbf{x}, \bs\beta) = \bs\beta^{\top} \mathbf{x}$, which is linear in $\mathbf{x}$. Therefore, the partial derivative of the true conditional quantile function with respect to $\mathbf{x}$ at the fixed evaluation point $(\mathbf{x}, \mathbf{h}_{\bs{s}_i}, \mathbf{w})$ is
\begin{equation*}
\frac{\partial g_{Y \mid (\mathbf{X}, \mathbf{H}, \mbW)}^{\tau} (\mathbf{x}, \mathbf{h}_{\bs{s}_i}, \mbw)}{\partial \mathbf{x}} = \frac{\partial}{\partial \mathbf{x}} \left[ \bs\beta^{\top} \mathbf{x} \cdot f_{12}(\mathbf{h}_{\bs{s}_i}, \mathbf{w}) \right] = f_{12}(\mathbf{h}_{\bs{s}_i}, \mathbf{w}) \cdot \bs\beta,
\end{equation*}
where we used the fact that $f_{12}(\mathbf{h}_{\bs{s}_i}, \mathbf{w})$ is constant with respect to $\mathbf{x}$.

This establishes the target quantity for Assumption~\ref{ass:1}. The key observation is that the partial derivative is a scalar multiple (determined by $f_{12}(\mathbf{h}_{\bs{s}_i}, \mathbf{w})$) of the parameter vector $\bs\beta$. Although $f_{12}(\mathbf{h}_{\bs{s}_i}, \mathbf{w})$ is unknown due to the latent confounder $\mathbf{h}_{\bs{s}_i}$, it remains fixed at location $\bs{s}_i$ for a given evaluation point $\mathbf{w}$.

\textbf{Step 3: Quantile regression estimation at location $\bs{s}_i$.}
Given $m_i$ time series observations $\{( \mathbf{X}_{\bs{s}_i}^t, Y_{\bs{s}_i}^t,  \mathbf{W}_{\bs{s}_i}^t)\}_{t=1}^{m_i}$ at location $\bs{s}_i$, we estimate the $\tau$-th conditional quantile function $\hat{g}^{m_i, \tau}_{Y \mid (\mathbf{X}, \mbW)}$ via quantile regression.

At a fixed evaluation point $(\mathbf{x}, \mathbf{w})$, the estimated quantile function is $\hat{g}^{m_i, \tau}_{Y \mid (\mathbf{X}, \mbW)}(\mathbf{X}_{\bs{s}_i}^{m_i}, \bs{Y}_{\bs{s}_i}^{m_i}, \bs{W}_{\bs{s}_i}^{m_i}) (\mathbf{x}, \mbw)$, which depends on all observed data $\{( \mathbf{X}_{\bs{s}_i}^t,  Y_{\bs{s}_i}^t, \mathbf{W}_{\bs{s}_i}^t)\}_{t=1}^{m_i}$ at location $\bs{s}_i$.

Since $f_{11}(\mathbf{x}, \bs\beta) = \bs\beta^{\top} \mathbf{x}$ is linear in $\mathbf{x}$ by Assumption C.4, and $f_{12}(\mathbf{h}_{\bs{s}_i}, \mathbf{w})$ is constant with respect to $\mathbf{x}$ (as both $\mathbf{h}_{\bs{s}_i}$ and $\mathbf{w}$ are fixed), the partial derivative of the true quantile function is
\begin{equation*}
\frac{\partial g_{Y \mid (\mathbf{X}, \mathbf{H}, \mbW)}^{\tau} (\mathbf{x}, \mathbf{h}_{\bs{s}_i}, \mbw)}{\partial \mathbf{x}} = f_{12}(\mathbf{h}_{\bs{s}_i}, \mbw) \cdot \bs\beta.
\end{equation*}

Similarly, the partial derivative of the estimated quantile function takes the form
\begin{equation*}
\frac{\partial \hat{g}^{m_i, \tau}_{Y \mid (\mathbf{X}, \mbW)}(\mathbf{X}_{\bs{s}_i}^{m_i}, \bs{Y}_{\bs{s}_i}^{m_i}, \bs{W}_{\bs{s}_i}^{m_i}) (\mathbf{x}, \mbw)}{\partial \mathbf{x}} = \hat{f}_{12}^{m_i}(\mathbf{h}_{\bs{s}_i}, \mbw) \cdot \hat{\bs\beta}^{m_i},
\end{equation*}
where $\hat{\bs\beta}^{m_i}$ is the quantile regression estimator for $\bs\beta$ and $\hat{f}_{12}^{m_i}(\mathbf{h}_{\bs{s}_i}, \mbw)$ represents the estimated scaling factor.

\textbf{Step 4: Establishing convergence.}
To prove Assumption~\ref{ass:1}, we need to show that 
\begin{equation*}
\frac{\partial \hat{g}^{m_i, \tau}_{Y \mid (\mathbf{X}, \mbW)}(\mathbf{X}_{\bs{s}_i}^{m_i}, \bs{Y}_{\bs{s}_i}^{m_i}, \bs{W}_{\bs{s}_i}^{m_i}) (\mathbf{x}, \mbw)}{\partial \mathbf{x}} - \frac{\partial g_{Y \mid (\mathbf{X}, \mathbf{H}, \mbW)}^{\tau} (\mathbf{x}, \mathbf{h}_{\bs{s}_i}, \mbw)}{\partial \mathbf{x}} \overset{\PP}{\rightarrow} \mathbf{0}
\end{equation*}
as $m_i \rightarrow \infty$.

Since $\mathbf{h}_{\bs{s}_i}$ is fixed (though unobserved) at location $\bs{s}_i$ and $\mathbf{w}$ is a fixed evaluation point, $f_{12}(\mathbf{h}_{\bs{s}_i}, \mbw)$ is a constant. Denote this constant as $c_{\bs{s}_i, \mathbf{w}} := f_{12}(\mathbf{h}_{\bs{s}_i}, \mbw)$. Then the convergence statement becomes
\begin{equation*}
\hat{f}_{12}^{m_i}(\mathbf{h}_{\bs{s}_i}, \mbw) \cdot \hat{\bs\beta}^{m_i} - c_{\bs{s}_i, \mathbf{w}} \cdot \bs\beta \overset{\PP}{\rightarrow} \mathbf{0}.
\end{equation*}

By the multiplicative separability and linearity in Assumption C.4, the quantile regression model at location $\bs{s}_i$ can be written as
\begin{equation*}
Q_\tau(Y \mid \mathbf{X}, \mathbf{W} = \mathbf{w}) = \hat{c}_{\bs{s}_i, \mathbf{w}}^{m_i} \cdot (\hat{\bs\beta}^{m_i})^{\top} \mathbf{X},
\end{equation*}
where $\hat{c}_{\bs{s}_i, \mathbf{w}}^{m_i} := \hat{f}_{12}^{m_i}(\mathbf{h}_{\bs{s}_i}, \mbw)$. Since $c_{\bs{s}_i, \mathbf{w}}$ and $\bs\beta$ appear multiplicatively, and $c_{\bs{s}_i, \mathbf{w}}$ depends on the unobserved $\mathbf{h}_{\bs{s}_i}$, we absorb the constant into the parameter estimate. Define $\tilde{\bs\beta}_{\bs{s}_i, \mathbf{w}} := c_{\bs{s}_i, \mathbf{w}} \cdot \bs\beta$, which is the identifiable quantity from the data.

The quantile regression estimator $\hat{\tilde{\bs\beta}}_{\bs{s}_i, \mathbf{w}}^{m_i}$ satisfies
\begin{equation*}
\hat{\tilde{\bs\beta}}_{\bs{s}_i, \mathbf{w}}^{m_i} = \argmin_{\tilde{\bs\beta}} \sum_{t=1}^{m_i} \rho_\tau\left(Y_{\bs{s}_i}^t - \tilde{\bs\beta}^{\top} \mathbf{X}_{\bs{s}_i}^t\right),
\end{equation*}
where $\rho_\tau(u) = u(\tau - \mathbb{I}\{u < 0\})$ is the pinball loss function. 

The consistency $\hat{\tilde{\bs\beta}}_{\bs{s}_i, \mathbf{w}}^{m_i} \overset{\PP}{\rightarrow} \tilde{\bs\beta}_{\bs{s}_i, \mathbf{w}} = c_{\bs{s}_i, \mathbf{w}} \cdot \bs\beta$ follows from Assumptions C.1--C.7 via the arguments in \citet{oberhofer2016asymptotic} (Appendix: Consistency, p.~701).

By the continuous mapping theorem, since
\begin{equation*}
\frac{\partial \hat{g}^{m_i, \tau}_{Y \mid (\mathbf{X}, \mbW)}(\mathbf{X}_{\bs{s}_i}^{m_i}, \bs{Y}_{\bs{s}_i}^{m_i}, \bs{W}_{\bs{s}_i}^{m_i}) (\mathbf{x}, \mbw)}{\partial \mathbf{x}} = \hat{\tilde{\bs\beta}}_{\bs{s}_i, \mathbf{w}}^{m_i}
\end{equation*}
and 
\begin{equation*}
\frac{\partial g_{Y \mid (\mathbf{X}, \mathbf{H}, \mbW)}^{\tau} (\mathbf{x}, \mathbf{h}_{\bs{s}_i}, \mbw)}{\partial \mathbf{x}} = \tilde{\bs\beta}_{\bs{s}_i, \mathbf{w}},
\end{equation*}
we have
\begin{equation*}
\frac{\partial \hat{g}^{m_i, \tau}_{Y \mid (\mathbf{X}, \mbW)}(\mathbf{X}_{\bs{s}_i}^{m_i}, \bs{Y}_{\bs{s}_i}^{m_i}, \bs{W}_{\bs{s}_i}^{m_i}) (\mathbf{x}, \mbw)}{\partial \mathbf{x}} - \frac{\partial g_{Y \mid (\mathbf{X}, \mathbf{H}, \mbW)}^{\tau} (\mathbf{x}, \mathbf{h}_{\bs{s}_i}, \mbw)}{\partial \mathbf{x}} \overset{\PP}{\rightarrow} \mathbf{0}
\end{equation*}
as $m_i \rightarrow \infty$, thereby verifying Assumption~\ref{ass:1}.
\end{proof}

\subsection{Discussion of Assumptions}
\begin{hyphenrules}{nohyphenation}
\begin{itemize}
\item \textbf{Assumption C.1} ensures the quantile condition $Q_\tau(\vartheta_t) = 0$ holds, which is critical for identifying the quantile function from the structural equation.
\item \textbf{Assumptions C.2--C.3} provide the necessary temporal dependence and regularity conditions on the error process $\{\vartheta_t\}$ to apply asymptotic theory for dependent data.
\item \textbf{Assumption C.5} (Identifiable Uniqueness) ensures that $\tilde{\bs\beta}_{\bs{s}_i, \mathbf{w}}$ is uniquely identified by requiring sufficient separation between different parameter values in terms of their fitted values, even when $f_{12}(\mathbf{h}_{\bs{s}_i}, \mbw)$ scales the relationship.
\item \textbf{Assumption C.6} (Moment Boundedness) prevents the regression function $f_{11}(\mathbf{x}_t, \bs\beta) \cdot f_{12}(\mathbf{h}_t, \mathbf{w}_t)$ from growing too rapidly, ensuring that sample averages converge.
\item \textbf{Assumption C.7} provides the necessary smoothness to apply uniform laws of large numbers over $\bs\Omega$.
\end{itemize}
\end{hyphenrules}

\baselineskip=12pt
\bibliographystyleapp{agsm}
\bibliographyapp{ref}

\end{appendix}

\end{document}